\documentclass[10pt]{article}
\pdfoutput=1

\usepackage{amsmath,amssymb,stmaryrd,xspace,enumerate}
\usepackage{tikz}
\usetikzlibrary{matrix}
\usetikzlibrary{decorations.markings}
\usetikzlibrary{shapes.geometric}
\usetikzlibrary{arrows,positioning} 
\usetikzlibrary{calc}
\usetikzlibrary{arrows.meta} 
\usetikzlibrary{matrix,shapes}
\usepackage[all,cmtip]{xy} 
\usepackage{wasysym}
\usepackage{tikz-cd}
\usepackage{mathtools}
\usepackage[T1]{fontenc} 
\usepackage{lmodern} 
\usepackage{multirow}
\usepackage{amsthm}
\usepackage{bbm}
\usepackage{geometry}
\usepackage{abstract}
\usepackage{hyperref}

\theoremstyle{definition}
\newtheorem{theorem}{Theorem}
\newtheorem{corollary}[theorem]{Corollary}
\newtheorem{lemma}[theorem]{Lemma}

\newtheorem{definition}[theorem]{Definition}

\newtheorem*{definition*}{Definition}
\newtheorem*{remark*}{Remark}
\newtheorem*{lemma*}{Lemma}
\newtheorem*{corollary*}{Corollary}
\newtheorem*{theorem*}{Theorem}

\newtheorem*{examples*}{Examples}
\newtheorem*{keyexample*}{Example} 
\newtheorem{assumption}{Assumption}
\newtheorem{axiom}{Axiom}

\newcommand{\id}[1]{\ensuremath{\mathrm{id}_{#1}}}

\newcommand{\catC}{\cat{C}}
\newcommand{\catB}{\cat{B}}

\newcommand\pfun{\mathrel{\ooalign{\hfil$\mapstochar\mkern5mu$\hfil\cr$\to$\cr}}}

\newcommand{\coproj}{\kappa}
\newcommand{\OTC}{\cat{(Mon)OTC}}
\newcommand{\MonOTC}{\OTC}

\newcommand{\OTCPlusStrict}{\cat{(Mon)OTC^{\otcplus}_{\text{strict}}}}

\newcommand{\Partial}{\text{Par}}
\newcommand{\lift}[1]{ \Partial({ {#1} }) } 

\DeclareFontFamily{U}{mathb}{\hyphenchar\font45}
\DeclareFontShape{U}{mathb}{m}{n}{
      <5> <6> <7> <8> <9> <10> gen * mathb
      <10.95> mathb10 <12> <14.4> <17.28> <20.74> <24.88> mathb12
      }{}
\DeclareSymbolFont{mathb}{U}{mathb}{m}{n}
\DeclareFontSubstitution{U}{mathb}{m}{n}
\DeclareMathSymbol{\mylgroup}{\mathbin}{mathb}{'160}
\DeclareMathSymbol{\myrgroup}{\mathbin}{mathb}{'161}

\newcommand{\totaspar}[1]{ \mylgroup {#1} \myrgroup } 

\newcommand{\ClassDet}{\mathsf{ClassDet}}
\newcommand{\ClassProb}{\mathsf{ClassProb}}

\newcommand{\Mat}{\mathsf{Mat}}
\newcommand{\FinHilbOTC}{\mathsf{FinHilb}} 
\newcommand{\FinCStarOTC}{\mathsf{FinCStar}} 
\newcommand{\CStarOTC}{\mathsf{CStar}} 
\newcommand{\Kl}{\cat{Kl}} 
\newcommand{\FDimCStar}{\cat{FdimCStar}_{\text{u}}^{\op}}

\newcommand{\FDimCStarBracket}{\cat{(Fdim)CStar}_{\text{u}}^{\op}}
\newcommand{\FDimCStarSUBracket}{\cat{(Fdim)CStar}_{\text{su}}^{\op}}

\newcommand{\total}{\text{total}} 

\newcommand{\Events}{\cat{Event}} 
\newcommand{\DetEvents}{\cat{DetEvent}} 

\newcommand{\eqdef}{\vcentcolon=}

\newcommand{\monopequiv}{\simeq_{\otimes}}
\newcommand{\monoropequiv}{\simeq_{(\otimes)}}

\newcommand{\otcplus}{+} 
\newcommand{\OTCFrom}[1]{\mathsf{OT}({#1})} 

\newcommand{\circlearrow}[1]{\ensuremath{\mathrlap{\hspace{2.5pt}\to}{\hspace{6pt}{#1}\hspace{6pt}}}}

\newcommand{\partialto}{\circlearrow{\partialsmallarrow[1.2pt]}}
\newcommand{\partialcomp}{\hspace{3.5pt} \partialsymbolcomp \hspace{2.5pt}}
\newcommand{\tikzcdpartialright}[1]{\arrow{r}[description, inner sep = 0]{\partialsymbollargearrow}[above]{{#1}}}

\newcommand{\totalto}{\circlearrow{\totalsmallarrow[1pt]}}
\newcommand{\eventcomp}{\circ}

\newcommand{\catopscomp}{\circ}
\newcommand{\catpopscomp}{\partialcomp}
\newcommand{\catopsto}{\totalto}
\newcommand{\catpopsto}{\partialto}

\newcommand{\partialsmallarrow}[1][1.2pt]{ 
\begin{tikzpicture}[baseline=-0.58ex]
\draw[black, fill=white,radius={#1}] (0,0) circle ; 
\draw[black, fill=white,radius=0.1pt] (0,0) circle ; 
\end{tikzpicture}
}

\newcommand{\totalsmallarrow}[1][1.2pt]{ 
\begin{tikzpicture}[baseline=-0.58ex]
\draw[black, fill=black,radius={#1}] (0,0) circle ; 
\end{tikzpicture}
}

\newcommand{\partialsymbolcomp}{ 
\begin{tikzpicture}[baseline=-0.5ex]
\draw[black, fill=white,radius=1.8pt] (0,0) circle ; 
\draw[black, fill=white,radius=0.3pt] (0,0) circle ; 
\end{tikzpicture}
}

\newcommand{\totalsymbolcomp}[1][1.4pt]{ 
\begin{tikzpicture}[baseline=-0.5ex]
\draw[black, fill=black,radius=1.5pt] (0,0) circle ; 
\end{tikzpicture}
}

\newcommand{\eventsymbolcomp}[1][1.4pt]{ 
\begin{tikzpicture}[baseline=-0.5ex]
\draw[black, fill=gray, radius=1.5pt] (0,0) circle ; 
\end{tikzpicture}
}

\newcommand{\partialsymbollargearrow}{ 
\begin{tikzpicture}[baseline=-0.6pt]
\draw[black, fill=white,radius=1.6pt] (0,0) circle ; 
\draw[black, fill=white,radius=0.1pt] (0,0) circle ; 
\end{tikzpicture}
}

\newcommand{\totalsymbollargearrow}{ 
\begin{tikzpicture}[baseline=-0.6pt]
\draw[black, fill=black,radius=1pt] (0,0) circle ; 
\end{tikzpicture}
}

\newcommand{\eventto}{
  \hspace{2.4pt}
  \begin{tikzpicture}[baseline=-0.58ex, inner sep = 0]
\draw[-{Triangle[scale=1.1,angle'=45, open]}] (0,0) -- (0.35,0);
  \end{tikzpicture}
      \hspace{1.2pt}
}

\makeatletter
\def\th@plain{%
  \thm@notefont{}
  \itshape 
}
\def\th@definition{%
  \thm@notefont{}
  \normalfont 
}
\makeatother

\tikzset{
    >=stealth',
    punkt/.style={
           rectangle,
           rounded corners,
           draw=black, 
           text width=6.5em,
           minimum height=2em,
           text centered},
    pil/.style={
           ->,
           shorten <=2pt,
           shorten >=2pt,},
     incl/.style={
           left hook->,
           shorten <=2pt,
           shorten >=2pt,},
     incl2/.style={
           right hook->,
           shorten <=2pt,
           shorten >=2pt,},
}

\DeclareFontFamily{U}{mathux}{\hyphenchar\font45}
\DeclareFontShape{U}{mathux}{m}{n}{
      <5> <6> <7> <8> <9> <10>
      <10.95> <12> <14.4> <17.28> <20.74> <24.88>
      mathux10
      }{}
\DeclareSymbolFont{mathux}{U}{mathux}{m}{n}

\DeclareMathSymbol{\bigovee}{1}{mathux}{"8F}
\DeclareMathSymbol{\bigperp}{1}{mathux}{"4E}

\newcommand{\Event}{\cat{Event}} 

\newcommand{\HilbH}{\mathcal{H}}
\newcommand{\HilbK}{\mathcal{K}}
\newcommand{\curlyB}{\mathcal{B}}

\newcommand{\op}{\ensuremath{\textrm{\rm op}}}
\newcommand{\cat}[1]{\ensuremath{\mathbf{#1}}\xspace}

\newcommand{\Set}{\cat{Set}}
\newcommand{\Kleisli}[1]{\text{Kl}(#1)}

\newcommand{\discard}[1]{\ensuremath{\tinyground_{#1}}}

\tikzstyle{cdiag}=[matrix of math nodes, row sep=3em, column sep=3em, text height=1.5ex, text depth=0.25ex,inner sep=0.5em]
\tikzstyle{arrow above}=[transform canvas={yshift=0.5ex}]
\tikzstyle{arrow below}=[transform canvas={yshift=-0.5ex}]

\newcommand{\vkcleararrows}{%
\tikzstyle{vkarrow1}=[arrow plain]
\tikzstyle{vkarrow2}=[arrow plain]
\tikzstyle{vkarrow3}=[arrow plain]
\tikzstyle{vkarrow4}=[arrow plain]
\tikzstyle{vkarrow5}=[arrow plain]
\tikzstyle{vkarrow6}=[arrow plain]
\tikzstyle{vkarrow7}=[arrow plain]
\tikzstyle{vkarrow8}=[arrow plain]
\tikzstyle{vkarrow9}=[arrow plain]
\tikzstyle{vkarrow10}=[arrow plain]
\tikzstyle{vkarrow11}=[arrow plain]
\tikzstyle{vkarrow12}=[arrow plain]}
\vkcleararrows

\tikzstyle{dot}=[circle, draw=black, fill=black!25, inner sep=.4ex]
\tikzstyle{black_dot}=[dot, fill=black!50]
\tikzstyle{white_dot}=[dot, fill=white]

\newif\ifvflip\pgfkeys{/tikz/vflip/.is if=vflip}
\newif\ifhflip\pgfkeys{/tikz/hflip/.is if=hflip}
\newif\ifhvflip\pgfkeys{/tikz/hvflip/.is if=hvflip}
\newlength\morphismheight
\newlength\wedgewidth
\tikzset{width/.initial=1mm}
\makeatletter
\tikzstyle{morphism}=[font=\small,morphismshape]
\pgfdeclareshape{morphismshape}
{
    \savedanchor\centerpoint
    {
        \pgf@x=0pt
        \pgf@y=0pt
    }
    \anchor{center}{\centerpoint}
    \anchorborder{\centerpoint}
    \saveddimen\overallwidth
    {
        \pgfkeysgetvalue{/tikz/width}{\minwidth}
        \pgf@x=\wd\pgfnodeparttextbox
        \ifdim\pgf@x<\minwidth
            \pgf@x=\minwidth
        \fi
    }
    \savedanchor{\upperrightcorner}
    {
        \pgf@y=.5\ht\pgfnodeparttextbox
        \advance\pgf@y by -.5\dp\pgfnodeparttextbox
        \pgf@x=.5\wd\pgfnodeparttextbox
    }
    \anchor{north}
    {
        \pgf@x=0pt
        \pgf@y=0.5\morphismheight
    }
    \anchor{north east}
    {
        \pgf@x=\overallwidth
        \multiply \pgf@x by 2
        \divide \pgf@x by 5
        \pgf@y=0.5\morphismheight
    }
    \anchor{east}
    {
        \pgf@x=\overallwidth
        \divide \pgf@x by 2
        \advance \pgf@x by 5pt
        \pgf@y=0pt
    }
    \anchor{west}
    {
        \pgf@x=-\overallwidth
        \divide \pgf@x by 2
        \advance \pgf@x by -5pt
        \pgf@y=0pt
    }
    \anchor{north west}
    {
        \pgf@x=-\overallwidth
        \multiply \pgf@x by 2
        \divide \pgf@x by 5
        \pgf@y=0.5\morphismheight
    }
    \anchor{south east}
    {
        \pgf@x=\overallwidth
        \multiply \pgf@x by 2
        \divide \pgf@x by 5
        \pgf@y=-0.5\morphismheight
    }
    \anchor{south west}
    {
        \pgf@x=-\overallwidth
        \multiply \pgf@x by 2
        \divide \pgf@x by 5
        \pgf@y=-0.5\morphismheight
    }
    \anchor{south}
    {
        \pgf@x=0pt
        \pgf@y=-0.5\morphismheight
    }
    \anchor{text}
    {
        \upperrightcorner
        \pgf@x=-\pgf@x
        \pgf@y=-\pgf@y
    }
    \backgroundpath
    {
    \begin{scope}
        \pgfkeysgetvalue{/tikz/fill}{\morphismfill}
        \pgfsetstrokecolor{black}
        \pgfsetlinewidth{.7pt}
        \begin{scope}
        \pgfsetstrokecolor{black}
        \pgfsetfillcolor{white}
        \pgfsetlinewidth{.7pt}
                \ifhflip
                    \pgftransformyscale{-1}
                \fi
                \ifvflip
                    \pgftransformxscale{-1}
                \fi
                \ifhvflip
                    \pgftransformxscale{-1}
                    \pgftransformyscale{-1}
                \fi
                \pgfpathmoveto{\pgfpoint
                    {-0.5*\overallwidth-5pt}
                    {0.5*\morphismheight}}
                \pgfpathlineto{\pgfpoint
                    {0.5*\overallwidth+5pt}
                    {0.5*\morphismheight}}
                \pgfpathlineto{\pgfpoint
                    {0.5*\overallwidth + \wedgewidth}
                    {-0.5*\morphismheight}}
                \pgfpathlineto{\pgfpoint
                    {-0.5*\overallwidth-5pt}
                    {-0.5*\morphismheight}}
                \pgfpathclose
                \pgfusepath{fill,stroke}
        \end{scope}
    \end{scope}
    }
}
\makeatother

\newenvironment{pic}[1][]
{\begin{aligned}\begin{tikzpicture}[font=\tiny,#1]}
{\end{tikzpicture}\end{aligned}}

\pgfdeclarelayer{foreground}
\pgfdeclarelayer{background}
\pgfdeclarelayer{morphismlayer}
\pgfsetlayers{background,main,morphismlayer,foreground}

\tikzstyle{dot}=[circle, draw=black, fill=black!20, inner sep=.4ex, node on layer=foreground]

\tikzstyle{whitedot}=[circle, draw=black, fill=white, inner sep=.4ex, node on layer=foreground]
\tikzstyle{greydot}=[circle, draw=black, fill=black!20, inner sep=.4ex, node on layer=foreground]
\tikzstyle{darkgreydot}=[circle, draw=black, fill=black!50, inner sep=.4ex, node on layer=foreground]
\tikzstyle{blackdot}=[circle, draw=black, fill=black, inner sep=.4ex, node on layer=foreground]

\tikzstyle{triangle} = [regular polygon, regular polygon sides=3, draw=black, fill=black!20,scale=0.3, node on layer=foreground]

\tikzstyle{whitetriangle}=[triangle, fill=white]
\tikzstyle{greytriangle}=[triangle, fill=black!20]
\tikzstyle{darkgreytriangle}=[triangle, fill=black!50]
\tikzstyle{blacktriangle}=[triangle, fill=black]

\tikzstyle{invertedtriangle} = [triangle,scale=-1]
\tikzstyle{whiteinvertedtriangle}=[invertedtriangle, fill=white]
\tikzstyle{greyinvertedtriangle}=[invertedtriangle, fill=black!20]
\tikzstyle{darkgreyinvertedtriangle}=[invertedtriangle, fill=black!50]
\tikzstyle{blackinvertedtriangle}=[invertedtriangle, fill=black]

\tikzstyle{morphism}=[font=\small,morphismshape]
\tikzstyle{box}=[rectangle,inner sep=.4ex, draw=black, node on layer=foreground]
\tikzstyle{mor}=[rectangle,inner sep=.4ex, draw=black, node on layer=foreground, minimum width = 0.6cm, minimum height = 0.4cm, font = \small]
\tikzstyle{morwide}=[mor, minimum width = 1.2cm]

\newif\ifvflip\pgfkeys{/tikz/vflip/.is if=vflip}
\newif\ifhflip\pgfkeys{/tikz/hflip/.is if=hflip}
\newif\ifhvflip\pgfkeys{/tikz/hvflip/.is if=hvflip}
\tikzset{width/.initial=1mm}

\newlength\minimummorphismwidth
\newlength\stateheight
\newlength\minimumstatewidth
\newlength\connectheight
\tikzset{colour/.initial=white}

\tikzstyle{mixed}=[line width=.7pt]
\tikzstyle{pure}=[line width=.7pt]
\tikzset{diredge/.style={decoration={
  markings,
  mark=at position 0.525 with {\arrow{#1}}},postaction={decorate}}}
\tikzset{
    diredge/.default=>
}

\tikzset{diredgestart/.style={decoration={
  markings,
  mark=at position 4pt with {\arrow{#1}}},postaction={decorate}}}
\tikzset{
    diredgestart/.default=<
}

\tikzset{diredgeend/.style={decoration={
  markings,
  mark=at position 1 with {\arrow{#1}}},postaction={decorate}}}
\tikzset{
    diredgeend/.default=>
}

\makeatletter

\pgfdeclareshape{morphismshape}
{
    \savedanchor\centerpoint
    {
        \pgf@x=0pt
        \pgf@y=0pt
    }
    \anchor{center}{\centerpoint}
    \anchorborder{\centerpoint}
    \saveddimen\overallwidth
    {
        \pgfkeysgetvalue{/tikz/width}{\minwidth}
        \pgf@x=\wd\pgfnodeparttextbox
        \ifdim\pgf@x<\minwidth
            \pgf@x=\minwidth
        \fi
    }
    \savedanchor{\upperrightcorner}
    {
        \pgf@y=.5\ht\pgfnodeparttextbox
        \advance\pgf@y by -.5\dp\pgfnodeparttextbox
        \pgf@x=.5\wd\pgfnodeparttextbox
    }
    \anchor{north}
    {
        \pgf@x=0pt
        \pgf@y=0.5\morphismheight
    }
    \anchor{north east}
    {
        \pgf@x=\overallwidth
        \multiply \pgf@x by 4
        \divide \pgf@x by 5
        \pgf@y=0.5\morphismheight
    }
    \anchor{east}
    {
        \pgf@x=\overallwidth
        \divide \pgf@x by 2
        \advance \pgf@x by 5pt
        \pgf@y=0pt
    }
    \anchor{west}
    {
        \pgf@x=-\overallwidth
        \divide \pgf@x by 2
        \advance \pgf@x by -5pt
        \pgf@y=0pt
    }
    \anchor{north west}
    {
        \pgf@x=-\overallwidth
        \multiply \pgf@x by 4
        \divide \pgf@x by 5
        \pgf@y=0.5\morphismheight
    }
    \anchor{south east}
    {
        \pgf@x=\overallwidth
        \multiply \pgf@x by 4
        \divide \pgf@x by 5
        \pgf@y=-0.5\morphismheight
    }
    \anchor{south west}
    {
        \pgf@x=-\overallwidth
       \multiply \pgf@x by 4
        \divide \pgf@x by 5
        \pgf@y=-0.5\morphismheight
    }
    \anchor{south}
    {
        \pgf@x=0pt
        \pgf@y=-0.5\morphismheight
    }
    \anchor{text}
    {
        \upperrightcorner
        \pgf@x=-\pgf@x
        \pgf@y=-\pgf@y
    }
    \backgroundpath
    {
    \begin{scope}
        \pgfkeysgetvalue{/tikz/fill}{\morphismfill}
        \pgfsetstrokecolor{black}
        \begin{scope}
        \pgfsetstrokecolor{black}
        \pgfsetfillcolor{white}
                \ifhflip
                    \pgftransformyscale{-1}
                \fi
                \ifvflip
                    \pgftransformxscale{-1}
                \fi
                \ifhvflip
                    \pgftransformxscale{-1}
                    \pgftransformyscale{-1}
                \fi
                \pgfpathmoveto{\pgfpoint
                    {-0.5*\overallwidth-5pt}
                    {0.5*\morphismheight}}
                \pgfpathlineto{\pgfpoint
                    {0.5*\overallwidth+5pt}
                    {0.5*\morphismheight}}
                \pgfpathlineto{\pgfpoint
                    {0.5*\overallwidth + \wedgewidth}
                    {-0.5*\morphismheight}}
                \pgfpathlineto{\pgfpoint
                    {-0.5*\overallwidth-5pt}
                    {-0.5*\morphismheight}}
                \pgfpathclose
                \pgfusepath{fill,stroke}
        \end{scope}
    \end{scope}
    }
}

\pgfdeclareshape{ground}
{
    \savedanchor\centerpoint
    {
        \pgf@x=0pt
        \pgf@y=0pt
    }
    \anchor{center}{\centerpoint}
    \anchorborder{\centerpoint}

    \anchor{north}
    {
        \pgf@x=0pt
        \pgf@y=0.16\stateheight
    }
    \anchor{south}
    {
        \pgf@x=0pt
        \pgf@y=0pt
    }
    \saveddimen\overallwidth
    {
        \pgfkeysgetvalue{/pgf/minimum width}{\minwidth}
        \pgf@x=\minimumstatewidth
        \ifdim\pgf@x<\minwidth
            \pgf@x=\minwidth
        \fi
    }
    \backgroundpath
    {
        \begin{pgfonlayer}{foreground}
        \pgfsetstrokecolor{black}
        \pgfsetlinewidth{1.25pt}
        \ifhflip
            \pgftransformyscale{-1}
        \fi
        \pgftransformscale{0.5}
        \pgfpathmoveto{\pgfpoint{-0.5*\overallwidth}{0pt}}
        \pgfpathlineto{\pgfpoint{0.5*\overallwidth}{0pt}}
        \pgfpathmoveto{\pgfpoint{-0.33*\overallwidth}{0.33*\stateheight}}
        \pgfpathlineto{\pgfpoint{0.33*\overallwidth}{0.33*\stateheight}}
        \pgfpathmoveto{\pgfpoint{-0.16*\overallwidth}{0.66*\stateheight}}
        \pgfpathlineto{\pgfpoint{0.16*\overallwidth}{0.66*\stateheight}}
        \pgfpathmoveto{\pgfpoint{-0.02*\overallwidth}{\stateheight}}
        \pgfpathlineto{\pgfpoint{0.02*\overallwidth}{\stateheight}}
        \pgfusepath{stroke}
        \end{pgfonlayer}
    }
}

\tikzstyle{groundwide}=[ground, minimum width = 1.2cm]

\makeatletter
\pgfkeys{%
  /tikz/on layer/.code={
    \pgfonlayer{#1}\begingroup
    \aftergroup\endpgfonlayer
    \aftergroup\endgroup
  },
  /tikz/node on layer/.code={
    \gdef\node@@on@layer{%
      \setbox\tikz@tempbox=\hbox\bgroup\pgfonlayer{#1}\unhbox\tikz@tempbox\endpgfonlayer\egroup}
    \aftergroup\node@on@layer
  },
  /tikz/end node on layer/.code={
    \endpgfonlayer\endgroup\endgroup
  }
}
\def\node@on@layer{\aftergroup\node@@on@layer}
\makeatother

\newcommand{\tinyground}[1][ground]{
\smash{\raisebox{-2pt}{\hspace{-3pt}\ensuremath{\begin{pic}[scale=0.4]
    \node[#1, scale=0.6] (1) at (0,0.4) {};
    \draw [pure] (1.south) to +(0,-.3);
\end{pic}
}\hspace{-1pt}}}}

\tikzset{stateshape/.style={append after command={
   \pgfextra
        \draw[sharp corners, fill=none]%
    (\tikzlastnode.west)%
    [rounded corners=0pt] |- (\tikzlastnode.north)%
    [rounded corners=0pt] -| (\tikzlastnode.east)%
    [rounded corners=5pt] |- (\tikzlastnode.south)%
    [rounded corners=5pt] -| (\tikzlastnode.west);
   \endpgfextra}}}

\tikzset{effectshape/.style={append after command={
   \pgfextra
        \draw[sharp corners, fill=none]%
    (\tikzlastnode.west)%
    [rounded corners=0pt] |- (\tikzlastnode.south)%
    [rounded corners=0pt] -| (\tikzlastnode.east)%
    [rounded corners=5pt] |- (\tikzlastnode.north)%
    [rounded corners=5pt] -| (\tikzlastnode.west);
   \endpgfextra}}}

\tikzstyle{state}=[stateshape,inner sep=.4ex, node on layer=foreground, minimum width = 0.6cm, minimum height = 0.4cm, font = \small]

\tikzstyle{effect}=[effectshape,inner sep=.4ex, node on layer=foreground, minimum width = 0.6cm, minimum height = 0.4cm, font = \small]

\tikzstyle{statewide}=[state, minimum width = 1.2cm]
\tikzstyle{stateextrawide}=[state, minimum width = 1.4cm]

\tikzstyle{effectwide}=[effect, minimum width = 1.2cm]

\newcommand{\CommaBin}{\mathbin{{, }}} 

\let\OLDthebibliography\thebibliography
\renewcommand\thebibliography[1]{
  \OLDthebibliography{#1}
  \setlength{\parskip}{4pt}
  \setlength{\itemsep}{0pt plus 0.3ex}
}

\title{}
\newcommand{\titlestyle}[1]{\textbf{#1}}

\newcommand{\indsys}[1]{\{#1\}}

\newcommand{\operations}{tests}
\newcommand{\operation}{test}

\newcommand{\Operations}{Tests}
\newcommand{\OpCat}{\cat{Test}} 
\newcommand{\ParOpCat}{\cat{ParTest}} 
\newcommand{\OpStructure}{\mathsf{Test}} 
\newcommand{\anoperation}{a test}

\newcommand{\unique}{{}!} 

\begin{document}

\begin{center}
\LARGE{ \titlestyle{ Operational Theories of Physics as Categories} }\\
\vspace{16pt}
\Large{ {Sean Tull} } \\ \vspace{3pt} 
\small {\nolinkurl{sean.tull@cs.ox.ac.uk}} \\ \vspace{3pt} 
\small{{University of Oxford, Department of Computer Science}} \\ \vspace{3pt} 
\end{center}

\begin{abstract}
We introduce a new approach to the study of operational theories of physics using category theory. We define a generalisation of the (causal) operational-probabilistic theories of Chiribella et al. and establish their correspondence with our new notion of an operational category. Our work is based on effectus theory, a recently developed area of categorical logic, to which we give an operational interpretation, demonstrating its relevance to the modelling of general probabilistic theories.
\end{abstract}

\section{Introduction}

Since the discovery that quantum systems may be used to easily perform tasks inherently difficult in the classical world \cite{deutsch1992rapid,shor1997polynomial}, there have been a host of approaches to understanding quantum theory in terms of the \emph{operations} it allows one to perform \cite{Hardy2001QTFrom5,Barrett2007InfoGPTs,PhysRevA.84.012311InfoDerivQT,abramsky2008categorical}. The general programme is to study quantum theory from among more general theories of physics, defined in terms of systems and, typically probabilistic, experiments one may perform upon them. Several surprising aspects of the quantum world, such as the famous \emph{no-cloning theorem} \cite{Barnum2007NoBroadcast}, have been found to in fact hold in all non-classical theories, while others such as \emph{quantum teleportation} have been found to be far more special \cite{barnum2012teleportation}.


Research initiated by Abramsky and Coecke~\cite{abramskycoecke:categoricalsemantics} has demonstrated the elegance of \emph{category theory} as a tool for describing the basic features of such operational theories. In particular, it is well-known that the physical events in any theory allowing systems to be placed `side-by-side' form a \emph{symmetric monoidal category} \cite{coecke2011categories}. While an entirely categorical approach to the study of quantum theory has been developed~\cite{abramsky2008categorical,coecke2010quantum}, by taking quantum features such as teleportation as primitive, it is no longer applicable to arbitrary probabilistic theories. 

The categorical and probabilistic approaches are combined in the framework of \emph{operational-probabilistic theories} due to Chiribella, D'Ariano and Perinotti~\cite{chiribella2010purification}, which forms the basis of a reconstruction of (finite-dimensional) quantum theory from purely operational principles \cite{PhysRevA.84.012311InfoDerivQT}. Here, a physical theory is associated with a (strict) symmetric monoidal category of physical events, each corresponding to a possible outcome of some experimental test allowed by the theory. Additional structure is then placed on top of this categories, specifying which events may form admissible tests, describing the classical data obtained in experiments, and allowing one to assign probabilities to experimental outcomes.

In this work, we present a new description of operational theories like these which allows us to treat them entirely categorically. Our approach is based on the use of \emph{coproducts} to model the classical data obtained in experiments, along with the related notion of \emph{coarse-graining}, and its use in forming \emph{controlled \operations{}}. This provides us with purely categorical descriptions of all of the primitive notions of (causal) operational-probabilistic theories, including derived concepts such as the ability to form convex combinations of physical events. 

The categorical structures we consider in our approach are not new, and were first described by Jacobs \cite{NewDirections2014aJacobs} as part of \emph{effectus theory} \cite{EffectusIntro}, a recently developed area of categorical logic for use in modelling quantum computation. Our results provide a new interpretation of effectus theory as being fundamentally of an operational nature, suggesting its use in describing computation in more general probabilistic theories, and allowing one to identify effectuses with operational theories satisfying certain basic extra axioms. Indeed, this work began life as an attempt to give effectus theory such an interpretation.

We begin in Section~\ref{sec:OTCs} by introducing the operational theories that we wish to study categorically, under the name of \emph{operational theories with control} (OTCs). These generalise the (causal) theories of Chiribella et al.\ \cite{chiribella2010purification} by allowing for more general `probabilities' than those simply in the unit interval $[0,1]$, but retain their key features such as the ability to coarse-grain over physical events. 

In Section~\ref{sec:OperationalCategories} we study the abstract properties of the category $\catB = \OpCat(\Theta)$ of \operations{} of such a theory $\Theta$, which provides our definition of an \emph{operational category}, a weakened form of Jacobs et al.'s notion of an effectus \cite{NewDirections2014aJacobs,StatesConv2015JWW}. 
We also consider the broader category $\catC = \ParOpCat(\Theta)$ of \emph{partial} \operations{}, i.e.\ tests which may yield no outcome at all, axiomatizing $\catC$ as an operational category \emph{in partial form}. We show that, in fact, each of these categories may be defined in terms of the other: any operational category $\catB$ defines an operational category in partial form $\Partial(\catB)$, while conversely, any such category $\catC$ defines an operational category $\catC_{\total}$ consisting of its `total' morphisms. This generalises an analogous correspondence central to effectus theory \cite{Partial2015Cho}. 

The `partial' category $\catC = \Partial(\catB)$ is crucial in Section~\ref{sec:OpCatsToOpTheories}, in which we show that each operational category defines an operational theory $\Theta = \OTCFrom{\catC}$ with $\catC$ as its category of physical events $\Events_{\Theta}$. The coproducts in $\catC$ endow this theory with the ability to form `direct sums' of its systems, and  we show that every operational theory $\Theta$ may be `completed' to a new one $\Theta^{\otcplus}$ of this form. Our first main result, Theorem~\ref{thm:MainEquivalence}, identifies operational categories with OTCs coming with such direct sums. 

Our definition of an OTC is deliberately chosen to be as weak as possible to allow this categorical treatment, and in Section~\ref{sec:AdditionalAssumptions} we discuss natural extra assumptions one might wish to add to our framework 
on purely operational grounds. In Section~\ref{sec:EffectusTheory} we finally establish the connections with effectus theory, with our next main result, Corollary~\ref{cor:Effiscetainotcs}, identifying those OTCs which correspond to effectuses. We can summarise the relations between operational theories and categories as follows.  
\begin{center}
\footnotesize
\begin{tikzpicture}[node distance=1.4cm, auto]

\node[punkt] (OTC) {Operational theories with control (OTCs)};
\node[punkt, right=of OTC] (OTCd) {OTCs with direct sums};
\node[punkt, right=of OTCd] (Par) {Operational categories in partial form};
\node[punkt, right=of Par] (Tot) {Operational categories (in total form)};

\node[punkt, below = 1.7cm of OTC] (OTCcomb) {OTCs satisfying Axioms~\ref{axiom:positivity}, \ref{axiom:combiningops}};
\node[punkt, right=of OTCcomb] (OTCdobs) {OTCs with direct sums satisfying Axioms~\ref{axiom:positivity}, \ref{axiom:obsdetops} }; 
\node[punkt,right=of OTCdobs] (effPar) {Effectuses in partial form};
\node[punkt, right=of effPar] (effTot) {Effectuses (in total form)}; 

\path (OTC) edge[pil, bend left=45] node[above]{$(-)^{\otcplus}$} (OTCd);
\path (OTCd) edge[incl] node[below]{} (OTC);  
\path (OTCd) edge[pil, bend left=45] node[above]{$\Events_{(-)}$} (Par); 
\path (Par) edge[pil, bend left=45] node[below]{$\OTCFrom{-}$} (OTCd); 
\path (Par) edge[pil, bend left=45] node[above]{$(-)_{\total}$} (Tot); 
\path (Tot) edge[pil, bend left=45] node[below]{$\Partial(-)$} (Par);

\path (OTCcomb) edge[pil, bend left=45] (OTCdobs);
\path (OTCdobs) edge[incl] node[below]{} (OTCcomb);  

\path (OTCcomb) edge[incl2] (OTC);
\path (OTCdobs) edge[incl2] (OTCd);
\path (effPar) edge[incl2] (Par);
\path (effTot) edge[incl2] (Tot);
\path (effPar) edge[pil, bend left=45] (effTot); 
\path (effTot) edge[pil, bend left=45] (effPar); 
\path (OTCdobs) edge[pil, bend left=45] (effPar); 
\path (effPar) edge[pil, bend left=45] (OTCdobs); 

\node[right=0.4cm of OTCd] {$\simeq$};
\node[right=0.4cm of Par]  {$\simeq$};
\node[right=0.4cm of OTCdobs] {$\simeq$};
\node[right=0.4cm of effPar]  {$\simeq$};

\end{tikzpicture}
 \vspace{5pt}\\
\end{center}
\normalsize
Finally, in Section~\ref{sec:CategoricalResults} we discuss how we can make the above picture precise, in terms of functors between the categories of operational theories and categories. 

\paragraph*{Notation.} For clarity we record our notation here for later reference. Throughout we use set-like notation $\{a_x\}_{x \in X}$ for what is really an $X$-indexed collection of entities $a_x$. We write $\to$, $\circ$ for arrows and composition in an arbitrary category, including operational categories in partial or total form. In the category $\Partial(\catB)$ associated to a category $\catB$ we instead use $\partialto, \partialcomp$. To each OTC $\Theta$ we associate the categories $\Events_{\Theta}$,  $\ParOpCat(\Theta)$ and $\OpCat(\Theta)$ of events, partial \operations{} and \operations{}, with arrows in each typically written $\eventto$, $\catpopsto$, $\catopsto$, respectively. 

As is usual, we write $\otimes$ for the tensor in a monoidal category. The symbol $\discard{}$ denotes discarding in an OTC or operational category in partial form. The notation $f \ovee g$ and $\bigoplus_{x \in X} A_x$, $\triangleright_y$ is used for coarse-graining and direct sum systems, respectively, in an OTC. In contrast, $f + g$, $[f,g]$, $A + B$, $\coprod_{x \in X} A_x$, $n \cdot A$, $\triangledown$, all refer to coproducts in a category (see Appendix~\ref{appendix:cattheory}). In an operational category in total or partial form, we again use the symbol $\triangleright$ for the `projections' $\triangleright_1: A + B \to A + 1, A$, respectively.


\section{Operational Theories with Control} \label{sec:OTCs}

\subsection{The framework}
Let us now describe what we will mean by an operational theory of physics. The theories we describe closely resemble the (causal) operational-probabilistic theories of Chiribella et al.\ \cite{chiribella2010purification}, retaining their key structural aspects, without explicitly assuming the use of traditional probabilities. We will introduce each of these features in turn, before giving the formal definition. 

\subsubsection{Systems, events and \operations{}}

In the operational approach to physics, we consider physical \emph{systems} $A, B, C \ldots$ and \emph{\operations{}} one may perform upon them. A test  $\{ {f_x \colon A \eventto B_x }\}_{x \in X}$ is to be thought of as a finite-outcome measurement one may perform on a system of type $A$. Any such \operation{} has a collection $X$ of classical \emph{outcomes}.
Though much of what we say works in the infinite case, we will always take these collections $X$ to be finite. 
 Each outcome $x \in X$ corresponds to an \emph{event} $f_x \colon A \eventto B_x$, which leaves us with a system of type $B_x$. We imagine that, on any given run of the \operation{}, precisely one the events $f_x \colon A \eventto B_x$ will occur, with the outcome $x \in X$ then recorded. 

It is part of the job of the theory to specify which finite collections $\{ {f_x \colon A \eventto B_x} \}_{x \in X}$ of events of the same input type form admissible \operations{}.
We will say that a finite collection $\{f_x \colon A \eventto B_x\}_{x \in X}$ forms a \emph{partial \operation{}} if they form a subset of some (total) \operation{} $\{f_y \colon A \eventto B_y\}_{y \in Y}$, with $X \subseteq Y$. We call a partial \operation{} which is in fact \anoperation{} \emph{total}. An event $f \colon A \eventto B$ is said to be \emph{deterministic} when $\{ f \colon A \eventto B \}$ is total. 

\subsubsection{The category of events}

In general, we think of an event $f \colon A \eventto B$ as a physical occurrence which transforms a system of type $A$ into one of type $B$. Given any two events $f \colon A \eventto B$ and $g \colon B \eventto C$, we may compose them to form a new event $g  \eventcomp f \colon A \eventto C$ 
interpreted as `$f$ occurs, and then $g$ occurs'. We assume that composition satisfies $h \eventcomp (g \eventcomp f) = (h \eventcomp g) \eventcomp f$, for all such composable triples, and that for each system $A$ there is an identity event $\id{A} \colon A \eventto A$, satisfying $f \eventcomp \id{A} = f$ and $\id{A} \eventcomp g = g$ for all $f \colon A \eventto B$, $g \colon C \eventto A$. In other words, the collection of events forms a \emph{category}. 

We always assume this category comes with a designated object $I$, called the `trivial' system, which we interpret as `nothing'. We call events $\omega \colon I \eventto A$, $e \colon A \eventto I$ and $r \colon I \eventto I$ \emph{states}, \emph{effects} and \emph{scalars}, respectively, and \operations{} of the form $\{e_x \colon A \eventto I \}_{x \in X}$ \emph{observations}.

\paragraph{Monoidal structure}
Along with the sequential composition $g \eventcomp f$ of events, we also typically assume a \emph{spatial} composition $\otimes$ which allows us to place systems and events `side-by-side'. Given any two systems $A$, $B$ we denote their composite system by $A \otimes B$. We may also compose any pair of events $f \colon A \eventto B$, $g \colon C \eventto D$ to form a new event $f \otimes g \colon A \otimes C \eventto B \otimes D$. In such a theory, it is often helpful to use the following graphical notation for events:
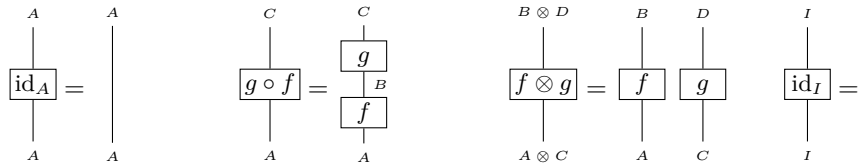
\[
  \begin{pic}
  \node[mor] (f) {$\id{A}$};
  \draw (f.south) to +(0,-.55) node[below] {$A$};
  \draw (f.north) to +(0,.55) node[above] {$A$};
  \end{pic}
  = 
  \begin{pic}
  \node[below](a) at (0,0){$A$};
  \node[above](b) at (0,1.55){$A$};
  \draw (a.north) to (b.south);
  \end{pic}
  \qquad \qquad
  \begin{pic}
  \node[mor] (f) {$g \eventcomp f$};
  \draw (f.south) to +(0,-.55) node[below] {$A$};
  \draw (f.north) to +(0,.55) node[above] {$C$};
  \end{pic}
  = 
  \begin{pic}
  \node[mor] (g) at (0,.75) {$g\vphantom{f}$};
  \node[mor] (f) at (0,0) {$f$};
  \draw (f.south) to +(0,-.2) node[below] {$A$};
  \draw (g.south) to node[right] {$B$} (f.north);
  \draw (g.north) to +(0,.2) node[above] {$C$};
  \end{pic}
  \qquad\qquad
  \begin{pic}
  \node[mor] (f) {$f \otimes g$};
  \draw (f.south) to +(0,-.55) node[below] {$A \otimes C$};
  \draw (f.north) to +(0,.55) node[above] {$B \otimes D$};
  \end{pic}
  = 
  \begin{pic}
  \node[mor] (f) at (-.4,0) {$f$};
  \node[mor] (g) at (.4,0) {$g\vphantom{f}$};
  \draw (f.south) to +(0,-.55) node[below] {$A$};
  \draw (f.north) to +(0,.55) node[above] {$B$};
  \draw (g.south) to +(0,-.55) node[below] {$C$};
  \draw (g.north) to +(0,.55) node[above] {$D$};
  \end{pic}
  \qquad
  \begin{pic}
  \node[mor] (f) {$\id{I}$};
  \draw (f.south) to +(0,-.55) node[below] {$I$};
  \draw (f.north) to +(0,.55) node[above] {$I$};
  \end{pic}
  = 
  \begin{pic}
  \end{pic}
\]
where in the final equation we mean that $\id{I}$ is given by the empty picture. 
We may then describe more complex events intuitively as graphical `circuits' such as: 
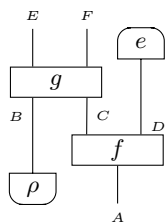
\[
\begin{pic}
\node[morwide](t1) at (0,1.9){$g$};
\node[morwide](t2) at (0.81,1.0){$f$};
\draw ([xshift=3mm]t1.south west) to node[left, yshift=2.5mm]{$B$}  +(0,-1) node[below,state](s) {$\rho$};
\draw ([xshift=-2mm]t1.south east) to ([xshift=2mm]t2.north west) node[right, yshift=2.5mm]{$C$};
\draw ([xshift=3mm]t1.north west) to +(0,0.5) node[above]{$E$};
\draw ([xshift=-2mm]t1.north east) to +(0,0.5) node[above]{$F$};
\draw (t2.south) to +(0,-0.5) node[below]{$A$};
\draw ([xshift=-3mm]t2.north east) to node[right, yshift=-4mm]{$D$}  +(0,1) node[above, effect]{$e$};
\end{pic}
\]

The circuit representation of events is fundamental in the operational reconstructions of quantum theory due to Hardy \cite{Hardy2011ReformQT} and the `Pavia group' \cite{PhysRevA.84.012311InfoDerivQT}. It is well-known that we may describe this situation precisely by requiring that our collection of events forms a \emph{symmetric monoidal} category. 
Any such category comes with `coherence' isomorphisms $\lambda_A  \colon I \otimes A \simeq A$, $\rho_A \colon A \otimes I \simeq A$, which encode the `triviality' of the system $I$, and $\alpha_{A,B,C} \colon A \otimes (B \otimes C) \simeq (A \otimes B) \otimes C$, $\sigma_{A,B} \colon A \otimes B \simeq B \otimes A$, which tell us that the order in which we compose systems is not significant. By working in the graphical language, one can in practice often avoid considering these isomorphisms altogether, and simply pretend we have equalities like $A \otimes I = A$ and $A \otimes (B \otimes C) = (A \otimes B) \otimes C$. The article \cite{coecke2011categories} provides an accessible introduction to the use of symmetric monoidal categories in physics. 

We will call a theory \emph{monoidal} when it comes with such a compositional structure. In any monoidal theory, we should be able to place any two \operations{} `side-by-side' to form a new one, as in the following.

\begin{assumption}[Parallel \operations{}] \label{assump_composition}
Whenever $\{f_x \colon A \eventto B_x\}_{x \in X}$ and $\{g_y \colon C \eventto D_y \}_{y \in Y}$ are \operations{}, so is 
 $\{f_x \otimes g_y \colon A \otimes C \eventto B_x \otimes D_y \}_{x \in X, y \in Y}$. 
\end{assumption}

Since monoidal structure is not essential to our approach, we will allow for non-monoidal theories. Typically each of our results are given in one form relevant to monoidal theories, and one to more general theories.

\subsubsection{Classical data flow} 

On top of the category of events, an operational theory also concerns the finite sets $X$ of classical outcome data associated with \operations{} $\{f_x\}_{x \in X}$. There are two main ways in which this data may be used. Firstly, we allow the outcome data of previous tests to be used as input in future ones. 

\begin{assumption}[Control] \label{assump:control} Given any \operation{} $\{ f_x \colon A \eventto B_x \}_{x \in X}$ and for each of its outcomes $x \in X$ \anoperation{} $\{{g^x}_y \colon B_x \eventto C_{x,y} \}_{y \in Y_x}$, the following forms a test: 
\[
\Big\{
\vcenter{ \hbox{
\begin{tikzpicture}
\node (a) at (0,0) {$A$};
\node (b) at (1.5,0) {$B_x$};
\node (c) at (3,0) {$C_{x,y}$};
\draw[-{Triangle[scale=1.1,angle'=45, open]}] (a) -- (b) node[midway,above]{\footnotesize $f_x$}; 
\draw[-{Triangle[scale=1.1,angle'=45, open]}] (b) -- (c) node[midway,above]{\footnotesize ${g^x}_y$}; 
\end{tikzpicture}
}}
\Big\}_{x \in X, y \in Y_x} 
\]
\end{assumption}

We refer to the above as a \emph{controlled test}, interpretting it as performing the test $\{f_x\}_{x \in X}$, and then depending on the outcome $x \in X$ choosing which \operation{} $g^x$ to perform next.
Control appears as an extra physical assumption in the framework in \cite{PhysRevA.84.012311InfoDerivQT}, which allows for theories without simple causal structure, and hence any straightforward notion of conditioning. 

Secondly, we assume that an agent is free to discard any amount of the classical data obtained in an experiment, thus `merging' several of its outcome events. Let us call a collection of events of the same type $\{f_x \colon A \eventto B \}_{x \in X}$ \emph{compatible} when they form a partial test. An operational theory should come with a rule for merging any compatible pair of events $\{f,g \colon A \eventto B\}$ into a new \emph{coarse-grained} event $f \ovee g \colon A \eventto B$, which we interpret as `$f$ occurs or $g$ occurs'.

\begin{assumption}[Coarse-graining] \label{assump:coarse_grain} The coarse-graining $f \ovee g$ of compatible pairs of events satisfies the following. 

\begin{itemize}
\item If $\{f_x\}_{x \in X} \cup \{ g, h \} $ is \anoperation{}, then so is $\{f_x\}_{x \in X} \cup \{g \ovee h\}$.

\item $(f \ovee g) \ovee h = f \ovee (g \ovee h)$ whenever $\{f, g, h\}$ are all compatible.

\item 
For compatible $\{g, h \colon B \eventto C\}$:
\begin{equation*} 
g \ovee h = h \ovee g, \qquad f \eventcomp (g \ovee h) = f \eventcomp g \ovee f \eventcomp h, \qquad (g \ovee h) \eventcomp k = g \eventcomp k \ovee h \eventcomp k 
\end{equation*} 
for all events $f \colon C \eventto D$, $k \colon A \eventto B$.

\item In a monoidal theory $f \otimes (g \ovee h) = (f \otimes g) \ovee (f \otimes h)$ for all $f$ and compatible $g$, $h$.

\end{itemize}
\end{assumption}

Each of these requirements has a straightforward operational interpretation. For example, the third equation above states that the event `either $g$ or $h$ occurs, then $f$ occurs' is the same as the event `either $g$ occurs then $f$ occurs, or $h$ occurs then $f$ occurs'.  Note that both sides of the above equations are indeed well-defined thanks to Assumptions~\ref{assump_composition} and \ref{assump:control}. The above lets us extend coarse-graining to finite, non-empty compatible collections by $\bigovee^n_{i = 1} {f_i} = f_1 \ovee (f_2 \ovee (\ldots \ovee f_n)$. 

\subsubsection{Remaining assumptions}

There are a few further assumptions we will need to add to our existing framework. Firstly, it will be helpful for us to assume for any two systems $A$, $B$ the presence of a special \emph{impossible event} $0_{A,B} \colon A \eventto B$ between them, which we interpret as the event which `never occurs'. We capture their expected behaviour as follows. 

\begin{assumption}[Impossible events] There is a family of events $0_{A,B} \colon A \eventto B$, defined over all pairs of systems $A$, $B$, such that:

\begin{itemize}
\item for all events $f \colon A \eventto B$ and $g \colon B \eventto C$ we have $g \eventcomp 0_{A,B} = 0_{A,C} = 0_{B,C} \eventcomp f $. In a monoidal theory we also have $f \otimes 0_{C,D} = 0_{A \otimes C, B \otimes D}$.
\item a collection of events $\{f_x \colon A \eventto B_x \}_{x \in X}$ forms \anoperation{} iff $\{f_x\}_{x \in X} \cup \{ 0_{A,C}\}$ does also, for any system $C$.
\item for all events $f \colon A \eventto B$ we have $f \ovee 0_{A,B} = f$.
\end{itemize}
\end{assumption}

Such a family of events is easily seen to be unique, if it exists. Since each $0_{A,B}$ is a unit for coarse-graining, we define the coarse-graining of each empty partial test to equal $0$. Next, we collect a couple of very basic assumptions that one would expect of any theory.


\begin{assumption}[Trivial \operations{}] \label{assump:TrivialDeterminism} Each identity event $\id{A}: A \eventto A$ is deterministic. In a monoidal theory, so is the coherence isomorphism $\lambda_I  = \rho_I \colon I \otimes I \simeq I$. 
\end{assumption}

\begin{assumption}[No superfluous events] \label{assump:EveryEventSomeOp}
Every event belongs to some \operation{}.
\end{assumption}

Our final assumption concerns effects $e \colon A \eventto I$, which are typically thought of as `yes-no' tests one can apply to a system. If we assume that every system comes with at least one observation, then our assumptions so far do in fact imply that each effect really does belong to some two-valued observation $\{e, e' \colon A \eventto I\}$. 
To see this, given some test $\{e \} \cup \{f_x: A \eventto B_x\}_{x \in X}$ and an observation $\{e^x_y\}_{y \in Y_x}$ on each system $B_x$, use control to define the observation $\{e\} \cup \{ \bigovee_{x, y} e^x_y \eventcomp f_x \}$. We require that such an observation always exists, and is unique. 
 
\begin{assumption}[Complementary effects] \label{assump:complementaryeffects} 
For every effect $e \colon A \eventto I$ there is a unique effect $e^{\bot} \colon A \eventto I$ for which $\{e, e^{\bot}\}$ forms a test.
\end{assumption}

Intuitively, $e$ may be tested to be true or false in any given state, and $e^{\bot}$ is then the unique event which occurs whenever `$e$ was found to be false'. We will see shortly that, like control, this assumption is closely related to our theory having a basic causal structure.

\begin{definition} \label{def:otc}
An \emph{operational theory with control} (OTC) $\Theta = (\Events_{\Theta}, I, \OpStructure, \ovee)$ is given by a category 
$\Events_{\Theta}$ of events, with distinguished object $I$, specification $\OpStructure$ of allowed \operations{}, and family of coarse-graining rules $\ovee$, satisfying Assumptions 2 to \ref{assump:complementaryeffects}. A \emph{monoidal} OTC further has $\Events_{\Theta}$ being symmetric monoidal with tensor unit $I$, and satisfies Assumption~\ref{assump_composition}.
\end{definition}

\subsection{Consequences of the assumptions}

The most crucial consequence of our choice of axioms for an OTC is the following.

\begin{lemma}[Causality]\label{lemmaCausality}  In an OTC, every system has a unique deterministic effect. 
\end{lemma}
\begin{proof}
An effect $e \colon A \eventto I$ is deterministic iff $\{e \colon A \eventto I, 0 \colon A \eventto I\}$ is total, i.e.\ iff $e = 0^{\bot}$.
\end{proof}

We denote the unique deterministic effect on a system $A$ by $\discard{A} \colon A \eventto I$, and call it \emph{discarding}, thinking of the single-outcome \operation{} $\{\discard{A} \colon A \eventto I\}$ as simply `throwing the system away'. In particular, by uniqueness we have $\discard{I} = \id{I}$. Causality appears as an explicit axiom in \cite{PhysRevA.84.012311InfoDerivQT}, where it is shown to ensure that tests performed in the future cannot affect the probabilities of tests performed in the past. Monoidal categories with such discarding maps $\discard{A} \colon A \eventto I$ have been studied in the context of causality by Coecke and Lal \cite{lal2012causal,coecke2014terminality}.  

\begin{lemma} \label{lemma:OTCConsequences} The following hold in any operational theory with control $\Theta$.
\begin{enumerate}[i)]
\item \label{lem_OTC_uniqueeffect} For every partial \operation{} $\{f_x \}_{x \in X}$ there is a unique effect $e$ for which $\{f_x\}_{x \in X} \cup \{ e\}$ is total.
\item A partial \operation{} $\{f_x \}_{x \in X}$ is total iff $\bigovee_{x \in X} \discard{} \eventcomp f_x= \discard{}$. In particular an event $f \colon A \eventto B$ is deterministic iff $\discard{B} \eventcomp f = \discard{A}$.
\item Coarse-graining $\ovee$ of effects is cancellative. That is, for all effects $a,b,c \colon A \eventto I$ we have $a \ovee b = a \ovee c \implies b =c$. In particular, scalar addition $\ovee$ is cancellative.
\item If $\Theta$ is a monoidal theory, then discarding satisfies $\discard{A \otimes B} = \lambda_I \eventcomp (\discard{A} \otimes \discard{B})$ where $\lambda_I \colon {I \otimes I} \eventto I$ is the coherence isomorphism.
Diagrammatically:
\[
\begin{pic}[scale=0.5]
\node[ground] (g) at (0,0) {};
\draw (g.south) to +(0,-0.75) node[below]{$A \otimes B$}; 
\end{pic}
=
\begin{pic}[scale=0.5]
\node[ground] (g1) at (-0.5,0){};
\node[ground] (g2) at (0.5,0){};
\draw (g1.south) to +(0,-0.75) node[below]{$A$}; 
\draw (g2.south) to +(0,-0.75) node[below]{$B$}; 
\end{pic}
\]
Further, the coherence isomorphisms $\rho_A$, $\lambda_A$, $\alpha_{A,B,C}$ and $\sigma_{A,B}$ are all deterministic.
\end{enumerate}
\end{lemma}
\begin{proof} 
We make free use of Assumptions 1 to \ref{assump:complementaryeffects}, along with causality.
\begin{enumerate}[i)]
\item 
If $\{f_x\}_{x \in X} \cup \{g_y\}_{y \in Y}$ is \anoperation{}, then so is $\{f_x\}_{x \in X} \cup \{\bigovee_{y \in Y} \discard{} \eventcomp g_y\}$, thanks to control and coarse-graining. Conversely, for any \operation{} $\{f_x\}_{x \in X} \cup \{e \colon A \eventto I\}$ we have $e = (\bigovee_{x \in X} \discard{} \eventcomp f_x)^{\bot}$, and so $e$ is unique.

\item As above, $e = (\bigovee_{x \in X} \discard{} \eventcomp f_x)$ is unique such that $\{f_x\}_{x \in X} \cup \{e^{\bot}\}$ forms a test. Then $\{f_x\}$ is total iff $e^{\bot} = 0$ iff $e = \discard{}$.




\item Suppose $a \ovee b = a \ovee c \colon A \eventto I$. By \ref{lem_OTC_uniqueeffect}) there are (unique) observations $\{a,b,d \}$, $\{a, c, e \}$. Then $d = (a \ovee b)^{\bot} = (a \ovee c)^{\bot} = e$, and so $b = (a \ovee d)^{\bot} = c$.

\item By Assumption~\ref{assump_composition}, \operations{} are closed under $\otimes$, and so the event $\discard{A} \otimes \discard{B} \colon A \otimes B \eventto I \otimes I$ is deterministic, for any $A$, $B$. By Assumption~\ref{assump:TrivialDeterminism}, $\lambda_I$ is deterministic, and hence by uniqueness we have $\discard{A \otimes B} = \lambda \eventcomp (\discard{A} \otimes \discard{B}) \colon A \otimes B \eventto I$. Each $\lambda_A$ is then deterministic by naturality, since we have
$\discard{A} \eventcomp \lambda_A = \lambda_I \eventcomp (\id{I} \otimes \discard{A}) = \discard{I \otimes A}$, and similarly so are the remaining coherence isomorphisms. 
\end{enumerate}
\end{proof}

Typical operational approaches to physics, such as \cite{chiribella2010purification}, come with the extra assumption that scalars $p \colon I \eventto I$ may be identified with probabilities $p \in [0,1]$. Here we do not make this assumption. Nonetheless, the scalars in an OTC behave in many ways like probabilities. They come with a partial addition $\ovee$, with identity element $0$, resembling the addition $p + q$ of probabilities $p,q \in [0,1]$ for which $p + q \leq 1$. 
There is a special scalar $1 = \discard{I}$, such that for each scalar $p$ there is a unique $p^{\bot}$ behaving like `$1-p$', with $p \ovee p^{\bot} = 1$. 
Further, we can multiply any two scalars as $p \bullet q = p \eventcomp q$, and in a monoidal theory, this multiplication is always commutative, with $p \eventcomp q = p \otimes q = q \circ p$ (the so-called `miracle of scalars' \cite{kelly1980coherence}). 

Using these features of our scalars, we can carry out probabilistic-style reasoning in an arbitrary OTC. 
For example, given any \operation{} $\{p_x \colon I \eventto I \}_{x \in X}$, thought of as a probability distribution, and any collection of states $\{\omega_x \colon I \eventto A\}_{x \in X}$, we define their \emph{convex combination} $\bigovee_{x \in X} p_x \bullet \omega_x$ to be the event:
\[
\bigovee_{x \in X} 
\big(
\begin{tikzcd}
I \rar{p_x} & I
\rar{\omega_x} & A 
\end{tikzcd}
\big)
\] 
which is well-defined thanks to the assumption of control.
In a monoidal theory we may define convex combinations of events of arbitrary type similarly. Using this notion of convexity, one may go on to define typical notions from operational physics such as \emph{completely mixed states}, and \emph{pure} events, reasoning much like in \cite{chiribella2010purification}. 


\subsection{Examples}

i) \emph{Deterministic, classical physics} is described by the OTC $\ClassDet$ in which systems are sets $A, B, C \ldots$ and events are partial functions $f \colon A \pfun{} B$. A collection of events $\{f_x \colon A \pfun B\}_{x \in X}$ forms \anoperation{} when their domains partition $A$, with coarse-graining $f \ovee g \colon A \pfun{} B$ given by disjoint union of functions. The monoidal structure comes from the usual Cartesian product $A \times B$ of sets, with trivial system $I = \{ \star \}$. 
The scalars in this theory are simply $\{0, 1\}$. \\

\noindent ii) In the theory $\ClassProb$ of \emph{probabilistic, classical physics}, events from $A$ to $B$ are functions sending each $a \in A$ to a finite probability subdistribution over elements of $B$. In other words, they are functions $f \colon A \times B \to [0,1]$ such that $\sum_{b \in B} f(a,b) \leq 1$, with only finitely many non-zero values $f(a,b)$, for each $a \in A$. A collection $\{f_x \colon A \eventto B_x\}_{x \in X}$ of such events forms \anoperation{} when
\[
\sum_{x \in X} \sum_{b \in B_x} f_x(a)(b) = 1
\]
holds for all $a \in A$. We may view such events $A \eventto B$ as `$A \times B$ matrices', in which each column has finitely many non-zero entries and sum $\leq 1$. Composition of events is then given by matrix composition, and coarse-graining by element-wise addition of matrices. The monoidal structure is again defined on systems by $A \times B = A \otimes B$, and on matrices using the usual Kronecker product. States $\omega \colon I \eventto A$ then correspond to finite probability distributions over elements of $A$, while effects simply $e \colon A \eventto I$ assign a probability $e(a)$ to each element $a \in A$. The scalars are given by probabilities $p \in [0,1]$.\\

\noindent iii) The most simple OTC describing quantum theory, $\FinHilbOTC$, takes systems to be finite-dimensional Hilbert spaces $\HilbH$, $\HilbK$. 
Events $f \colon \HilbH \eventto \HilbK$ are given by completely-positive maps $\bar{f} \colon \curlyB(\HilbK) \to \curlyB(\HilbH)$ between their spaces of bounded operators which are sub-unital, i.e.\ trace non-increasing. Note that we work in the Heisenberg picture, with maps in the `opposite direction'. 
A finite collection of events $\{f_x \colon \HilbH \eventto \HilbK_x\}_{x \in X}$ forms \anoperation{} when the completely positive map they induce from $\bigoplus_{x \in X} \curlyB(\HilbK_x)$ to $\curlyB(\HilbH)$ is unital i.e.\ trace-preserving. Coarse-graining is given by addition of completely positive maps, and
the monoidal structure $\otimes$ extends to events from the usual tensor product of Hilbert spaces, with trivial system $I = \mathbb{C}$.

As special cases, states $\omega \colon I \eventto \HilbH$ are then given by completely positive, trace non-increasing maps $\curlyB(\HilbH) \to \mathbb{C}$, which correspond precisely to (subnormalised) density matrices $\rho \in \curlyB(\HilbH)$ by Gleason's Theorem.  Effects, i.e.\ completely positive, sub-unital maps $e \colon \mathbb{C} \to \curlyB(\HilbH)$ correspond to effects on $\HilbH$ in the usual sense, namely positive operators $e \in \curlyB(\HilbH)$ satisfying $0 \leq e \leq 1$. Again, the collection of scalars is $[0,1]$. 

More generally, we can extend our notion of system to define OTCs $\FinCStarOTC$ and $\CStarOTC$ of finite-dimensional and arbitrary C*-algebras, respectively, in exactly the same way. \\

Our final example demonstrates that, despite the discussion above, scalars in an OTC can be quite different from probabilities in general. In Section~\ref{sec:AdditionalAssumptions} we discuss extra axioms one may add to our framework to rule out such examples.\\

\noindent iv) For any unital semiring $R$, let $R_{\leq 1} = \{a \in R \mid (\exists b \in R) \ \ \ a + b = 1 \}$. The OTC $\Mat_R$ takes systems to be natural numbers $n \in \mathbb{N}$, with events $M \colon n \eventto m$ given by $n \times m$ matrices with values in $R_{\leq 1}$. A finite collection $\{M_x \colon n \eventto m_x\}_{x \in X}$ forms \anoperation{} iff we have $\sum_{x \in X} \sum^{m_x}_{j = 1} M_x(i,j) = 1$ for all $i \in I$, while coarse-graining is given by elementwise addition of matrices. The scalars in $\Mat_R$ are $R_{\leq 1}$. For example, in $\Mat_{\mathbb{Z}}$ the scalars are the integers $\mathbb{Z}$.

\section{Operational Categories} \label{sec:OperationalCategories}

Our definition of an OTC was quite long, involving placing on top of the category of events the additional structure of allowed \operations{} and coarse-graining rules, along with our Assumptions. In fact, the \operations{} and partial \operations{} of $\Theta$ also form categories $\catB = \OpCat(\Theta)$ and $\catC = \ParOpCat(\Theta)$, each definable in terms of the other. We will find that these categories each provide a more elegant way to study the theory $\Theta$, with all of its crucial aspects encoded in their categorical properties.

\subsection{The category of \operations{}}

For any OTC $\Theta$, we define its category $\OpCat(\Theta)$ of tests, with morphisms denoted $\catopsto$, as follows:


\begin{itemize}
\item objects are finite, indexed collections $\indsys{A_x}_{x \in X}$ of systems of $\Theta$;
\item morphisms $f \colon \indsys{A} \catopsto \indsys{B_y}_{y \in Y}$ are \operations{} $\{f_y \colon A \eventto B_y\}_{y \in Y}$ in $\Theta$. More generally, a morphism  $M \colon \indsys{A_x}_{x \in X} \catopsto \indsys{B_y}_{y \in Y}$ is  an $X$-indexed collections of \operations{} $\{M(x,y) \colon A_x \eventto B_y\}_{y \in Y}$. 
\end{itemize}
These morphisms may be viewed as matrices of events $M(x,y) \colon A_x \eventto B_y$ for which each `column' is \anoperation{} in $\Theta$. Composition is then given by matrix composition, using coarse-graining:
\[
(N \circ M )(x,z) = \bigovee_{y \in Y} 
\Big(
\vcenter{ \hbox{ 
\begin{tikzpicture}[baseline=(current bounding box.center)]
\node (a) at (0,0) {$A_x$};
\node (b) at (1.5,0) {$B_y$};
\node (c) at (3,0) {$C_z$};
\draw[-{Triangle[scale=1.1,angle'=45, open]}] (a) -- (b) node[midway,above]{\scriptsize $M(x,y)$};
\draw[-{Triangle[scale=1.1,angle'=45, open]}] (b) -- (c) node[midway,above]{\scriptsize $N(y,z)$};
\end{tikzpicture}
}
}
\Big)
\]
for $M \colon \indsys{A_x}_{x \in X} \catopsto \indsys{B_y}_{y \in Y}$, $N \colon \indsys{B_y}_{y \in Y} \catopsto \indsys{C_z}_{z \in Z}$.
When $\Theta$ is a monoidal theory, $\OpCat(\Theta)$ is a symmetric monoidal category under 
$\indsys{A_x}_{x \in X} \otimes \indsys{B_y}_{y \in Y} = \indsys{A_x \otimes B_y}_{(x,y) \in X \times Y}$ and 
\begin{equation*} 
(M \otimes N)((x,w),(y,z)) = M(x,y) \otimes N(w,z) \colon A_x \otimes {B}_{w} \eventto C_y \otimes {D}_{z}
\end{equation*}
for morphisms $M \colon \indsys{A_x}_{x \in X} \catopsto \indsys{C_y}_{y \in Y}$ and $N \colon \indsys{B_{w}}_{w \in W} \catopsto \indsys{{D}_{z}}_{z \in Z}$.

\paragraph{Properties of $\OpCat(\Theta)$} 
We now wish to explore the properties of this category $\OpCat(\Theta)$. We introduce each of the basic categorical notions we require by example; Appendix~\ref{appendix:cattheory} contains more information for those new to these concepts. 


Firstly, causality implies that for every object $A = \indsys{A_x}_{x \in X}$, there is a unique arrow $\unique \colon A \catopsto \indsys{I}$, given by $ \unique(x) = \discard{} \colon A_x \eventto I$. This makes $\indsys{I}$ is a \emph{terminal object} in this category, which we denote by $1$. Similarly, there is trivially a unique arrow from the empty collection of systems $ 0 = \indsys{\text{ }}$ to each object $A$, making it an \emph{initial object}.


Now consider a general object $A = \indsys{A_x}_{x \in X}$ in $\OpCat(\Theta)$. For each $x \in X$ and singleton $\indsys{A_x}$ there is an arrow $\coproj_x \colon \indsys{A_x} \catopsto A$ given by the \operation{} $\{\id{} \colon A_x \eventto A_x \} \cup \{0 \colon A_x \eventto A_y \}_{x \neq y}$. These come with the property that, for any collection of arrows $M_x \colon \indsys{A_x} \catopsto B$, for $x \in X$, there is a unique arrow $M \colon A \catopsto B$ with $M \catopscomp \coproj_x = M_x$, for all $x$. This means precisely that $A$ forms a \emph{coproduct} $A = \coprod_{x \in X} \indsys{A_x}$ of the collection of objects $\indsys{A_x}$, for $x \in X$. Hence the category $\OpCat(\Theta)$ has coproducts $(+,0)$ of all finite collections of objects. 

The coproducts and terminal object are related by the following rule. 
Consider \anoperation{} $\{f_x \colon B \eventto A_x\}_{x \in X} \cup \{e \colon B \eventto I\}$ in $\Theta$, corresponding to an arrow $f \colon \indsys{B} \catopsto A + 1$ in $\OpCat(\Theta)$, where $A = \indsys{A_x}_{x \in X}$. When $\{f_x\}_{x \in X}$ is total, it corresponds to a unique arrow $g \colon \indsys{B} \catopsto A$, with $f$ then equal to $\coproj_1 \catopscomp g$. This occurs precisely when 
the effect $e$ is equal to $0$, i.e.\ when the morphisms $(\unique + \unique) \catopscomp f = \{ \bigovee_{x \in X} \discard{} \eventcomp f_x, e \}$ and $\coproj_1 \catopscomp \unique = \{\discard{}, 0_{B,I} \}$ are equal: 
\[
\begin{tikzcd}[font=\normalsize]
\indsys{B}
\arrow[bend left]{drr}[description, inner sep = 0]{\totalsymbollargearrow}[above]{\unique}
\arrow[bend right]{ddr}[description, inner sep = 0]{\totalsymbollargearrow}[swap]{\{f_x \colon B \eventto A_x\}_{x \in X} \cup \{e \colon B \eventto I\}}
\arrow[dr, dotted, "\exists \ \unique \ \{f_x\}_{x \in X}" description] & & \\
& \indsys{A_x}_{x \in X} 
\arrow{d}[description, inner sep = 0]{\totalsymbollargearrow}[]{\coproj_1}
\arrow{r}[description, inner sep = 0]{\totalsymbollargearrow}[]{\unique}
& \indsys{I} 
\arrow{d}[description, inner sep = 0]{\totalsymbollargearrow}[]{\coproj_1} \\
& \indsys{A_x}_{x \in X} + \indsys{I}
\arrow{r}[description, inner sep = 0]{\totalsymbollargearrow}[swap]{\unique{} + \unique{}}
& \indsys{I, I}
\end{tikzcd}
\]

We can summarise this by saying that the bottom-right square above is a \emph{pullback} in $\OpCat(\Theta)$. 
From now on, we assume a basic familiarity with these categorical notions, as provided by Appendix~\ref{appendix:cattheory}. 

\subsection{Operational categories, in total form}

We now reach our main definition, which abstractly describes the category $\OpCat(\Theta)$ of \operations{} of a (monoidal) OTC $\Theta$. All of the categorical structures here were first identified by Jacobs in \cite{NewDirections2014aJacobs}.

\begin{definition} \label{def:Opcat} An \emph{operational category} is a category $\catB$ with a terminal object $1$ and finite coproducts $(+, 0)$, for which for all objects $A$:
\begin{enumerate}[1)]
\item \label{opcatjointmonic} the morphisms $[\triangleright_1, \coproj_2], [\triangleright_2, \coproj_2] \colon  (A + A) + 1 \to A + 1$  are jointly monic, where we define $\triangleright_i \colon A + A \to A + 1$ by $\triangleright_1 = [\coproj_1, \coproj_2 \circ \unique]$ and $\triangleright_2 = [\coproj_2 \circ \unique, \coproj_1]$;

\item \label{opcatpullweak} the following diagram is a pullback:
\begin{tikzcd}
A \rar{\unique} \dar[swap, "\coproj_1"] & 1 \dar{\coproj_1} \\
A + 1 \rar[swap]{\unique + \unique} & 1 + 1
\end{tikzcd}
\end{enumerate}

\noindent A \emph{monoidal operational category} is in addition symmetric monoidal $(\catB, \otimes, I)$ with the terminal object as its tensor unit $1 = I$, and with the tensor $\otimes$ \emph{distributing over coproducts}, meaning that the canonical maps 
\begin{equation}\label{eq:distributive}
\begin{tikzcd}[column sep = large]
(A \otimes B) + (A \otimes C)  
\arrow[rr, "\lbrack \id{A} \otimes \coproj_1 \CommaBin \  \id{A} \otimes \coproj_2 \rbrack" ]
&&  A \otimes(B + C)
\end{tikzcd},  \qquad
\begin{tikzcd}
0 \rar{\unique} &  A \otimes 0 
\end{tikzcd}
\end{equation}
are isomorphisms.
\end{definition}

We will also sometimes call such a category a (monoidal) operational category \emph{in total form}. 
The `joint monicity' in condition \ref{opcatjointmonic}) means that whenever we have $[\triangleright_i, \coproj_2] \circ f = [\triangleright_i, \coproj_2] \circ g$, for $i = 1 ,2$, then $f = g$. In any distributive monoidal category, the isomorphisms above extend to arbitrary finite coproducts, with $A \otimes (\coprod_{x \in X} B_x) \simeq \coprod_{x \in X} (A \otimes B_x)$. As a special case, writing $n$ for the object $n \cdot 1 = \underbrace{1 + \ldots + 1}_n$, we have $A \otimes n \simeq n \cdot A$. In Appendix~\ref{appendix:propertiesOfOpCats}, we show that the pullback in condition \ref{opcatpullweak}) extends as follows:

\begin{lemma} \label{opcatlemma}
In an operational category, all coprojections $\coproj_1 \colon A \to A + B$ are monic, and diagrams of the following forms are pullbacks:
\[
\begin{tikzcd}
A \rar{f} \arrow[d, swap, "\coproj_1"] & B \dar{\coproj_1} \\
A + C \rar[swap]{f + \id{}} & B + C
\end{tikzcd}
\qquad
\qquad
\qquad
\begin{tikzcd}
0 \arrow[r, "\unique"] \arrow[d, "\unique", swap] & A \arrow[d, "\coproj_1"] \\
B \arrow[r, "\coproj_2", swap] & A + B
\end{tikzcd}
\]
\end{lemma}

\begin{keyexample*} As outlined above, for any (monoidal) OTC $\Theta$ the category $\OpCat(\Theta)$ is a (monoidal) operational category. In the monoidal case, by definition we have $\indsys{A} \otimes \indsys{B,C} = \indsys{A \otimes B, A \otimes C}$, from which distributivity follows. We have only yet to discuss the joint monicity condition \ref{opcatjointmonic}). To understand it, we will need to consider how $\OpCat(\Theta)$ in fact also captures the partial \operations{} of $\Theta$, to which we now turn.
\end{keyexample*}

\subsection{Operational categories, in partial form}

When working with an operational theory $\Theta$, it is helpful to consider not only \operations{}, but also more general partial \operations{}, including individual events. We define its category $\ParOpCat(\Theta)$ of partial \operations{}, with arrows denoted $\catpopsto$, just like $\OpCat(\Theta)$, but with morphisms $M \colon \indsys{A_x}_{x \in X} \catpopsto \indsys{B_y}_{y \in Y}$ now given more generally by $X$-indexed collections of partial \operations{}.

\paragraph{Properties of $\ParOpCat(\Theta)$} 
The empty collection $ 0 = \indsys{\text{ }}$ of systems is now a \emph{zero object}, meaning it is both initial and terminal, with the unique arrow $0 \colon { \indsys{A_x}_{x \in X} \catpopsto 0 \catpopsto \indsys{B_y}_{y \in Y} }$ between any two objects given by the zero matrix. $\ParOpCat(\Theta)$ has finite coproducts just like those of $\OpCat(\Theta)$, i.e.\ we have $\indsys{A_x}_{x \in X} = \coprod_{x \in X} \indsys{A_x}$, so that each partial \operation{} $\{f_x \colon A \eventto B_x\}_{x \in X}$ in $\Theta$ corresponds to a morphism $f \colon \indsys{A} \catpopsto \coprod_{x \in X} \indsys{B_x}$. Each individual event $f_y$ is then described categorically as the composite:  
\[
\begin{tikzcd}
\indsys{A} \tikzcdpartialright{f} & \coprod_{x \in X} \indsys{B_x} \tikzcdpartialright{\triangleright_y} & \indsys{B_y}
\end{tikzcd}
\]
where in any category with coproducts and a zero object we define `projections' ${\triangleright_y \colon \coprod_{x \in X} A_x \to A_y}$ by:
\begin{equation} \label{eq:projectionarrows_partial}
\triangleright_y \circ \coproj_x \quad = \quad  
\left\{\begin{array}{ll}
\id{} & \mbox{if } x = y \\
0 & \mbox{if } x \neq y
\end{array}\right.
\end{equation}
\noindent
(via $A + 0 \simeq A$, this coincides with our earlier definition of $\triangleright_i$ in $\catB$).
Since any partial \operation{} $f \colon \indsys{A} \catpopsto \coprod_{x \in X} \indsys{B_x}$ is determined entirely by its collection of events, together the maps $\triangleright_x$ are jointly monic. 

Finally, causality provides the extra structure of a family of `discarding' morphisms $\discard{A} \colon A \catpopsto I$, given on an object $ A = \indsys{A_x}_{x \in X}$ as before by $\discard{A}(x) = \discard{} \colon A_x \eventto I$. In any category $\catC$ with such specified morphisms, we will call an arrow $f \colon A \to B$ \emph{total} when it satisfies $\discard{B} \circ f = \discard{A}$, writing $\catC_{\total}$ for the subcategory of total arrows. By Lemma~\ref{lemma:OTCConsequences}, a partial \operation{} is indeed total in $\Theta$ iff it is total in this sense. We summarise these properties with our next main definition.

\begin{definition} \label{def:opcatpartial} An \emph{operational category in partial form}  $(\catC, I, \discard{})$ is a given by category $\catC$ with finite coproducts and a zero object $(+,0)$, together with a specified object $I$ and family of arrows $\discard{A} \colon A \to I$, such that for all objects $A$, $B$:
\begin{enumerate}[1)]
\item \label{paropcatatdiscardConds} $\discard{I} = \id{I}$ and $\discard{A + B} = [\discard{A}, \discard{B}] \colon A + B \to I$;
\item \label{paropcatcatJM} the maps $\triangleright_1, \triangleright_2 \colon A + A \to A$ are jointly monic;

\item \label{paropcatuniquetotal} for every $f \colon A \to B$ there is a unique total $g \colon A \to B + I$ with $f = \triangleright_1 \circ g$. 
\end{enumerate}

\noindent A \emph{monoidal operational category in partial form} is in addition symmetric monoidal $(\catC, \otimes, I)$, with the chosen object $I$ forming the tensor unit, such that:
\begin{enumerate}[1)]
\setcounter{enumi}{3}
\item \label{monoidalparopcatdistrib} the tensor $\otimes$ distributes over the coproducts;

\item \label{monoidalparopcatdiscard} $\discard{A \otimes B} = \lambda_I \circ (\discard{A} \otimes \discard{B}) \colon A \otimes B \to I$, where $\lambda_I \colon I \otimes I \to I$ is the coherence isomorphism. As before, this may be depicted:  
\[
\begin{pic}[scale=0.5]
\node[ground] (g) at (0,0) {};
\draw (g.south) to +(0,-0.75) node[below]{$A \otimes B$}; 
\end{pic}
=
\begin{pic}[scale=0.5]
\node[ground] (g1) at (-0.5,0){};
\node[ground] (g2) at (0.5,0){};
\draw (g1.south) to +(0,-0.75) node[below]{$A$}; 
\draw (g2.south) to +(0,-0.75) node[below]{$B$}; 
\end{pic}
\]
\end{enumerate}
\end{definition}

The stronger joint monicity condition we mentioned for $\ParOpCat(\Theta)$ holds more generally: 

\begin{lemma} \label{lemma:opcatpartialjointmon} \cite[Lemma~5]{EffectusIntro} In any operational category in partial form $\catC$, for any finite coproduct the projections $\triangleright_y \colon \coprod_{x \in X} A_x \to A_y$ are jointly monic.
\end{lemma}

\begin{keyexample*} For any (monoidal) OTC $\Theta$, $\catC = \ParOpCat(\Theta)$ is a (monoidal) operational category in partial form. We have only to discuss requirements \ref{paropcatuniquetotal}) and \ref{monoidalparopcatdiscard}), which were shown in Lemma~\ref{lemma:OTCConsequences}. 
\end{keyexample*}

\subsection{Equivalence of total and partial forms} \label{subsec:totparequiv}
 
In fact, the categories $\OpCat(\Theta)$ and $\ParOpCat(\Theta)$ may each be defined in terms of the other, and we will see further that this extends to more general operational categories, explaining our `partial form' terminology.

First, let's consider how $\OpCat(\Theta)$ in fact encodes the partial \operations{} of $\Theta$, including its individual events. We saw in Lemma~\ref{lemma:OTCConsequences} that any partial \operation{} $\{f_x \colon A \eventto B_x\}_{x \in X}$ may be equated with the \operation{} $\{f_x\}_{x \in X} \cup \{ e \colon A \eventto I\}$, where $e = (\bigovee_{x \in X} \discard{} \eventcomp f_x)^{\bot}$. Intuitively, we perform some \operation{} containing $\{f_x\}_{x \in X}$ and if none of the outcomes $x \in X$ are obtained, discard the system. Hence partial \operations{} $\indsys{A} \catpopsto \indsys{B_x}_{x \in X}$ may be identified with arrows $\indsys{A} \catopsto \indsys{B_x}_{x \in X} + 1$ in $\OpCat(\Theta)$. 

\paragraph{The category $\Partial(\catB)$} This situation of a `partial' category associated to a given `total' category $\catB$ has been studied already by Cho~\cite{Partial2015Cho} and Jacobs et al.\ \cite{EffectusIntro} in the context of effectus theory (see Section~\ref{sec:EffectusTheory}), and we borrow their approach here. For any category $\catB$ with finite coproducts $(+,0)$ and a terminal object $1$, by a \emph{partial} arrow $f \colon A \partialto B$ we mean an arrow $f \colon A \to B + 1$ in $\catB$. These partial arrows form a category $\Partial(\catB)$ under composition:

\begin{equation*} 
\big(
\begin{tikzcd}
A \tikzcdpartialright{f} & B \tikzcdpartialright{g} & C
\end{tikzcd}
\big)
=
\big(
\begin{tikzcd}
A \rar{f} & B + 1 \rar{[g,\coproj_2]} & C + 1
\end{tikzcd}
\big)
\end{equation*}
\noindent which we denote by $g \partialcomp f$, with the identity $\id{A}  \colon A \partialto A$ given by $\coproj_1 \colon A \to A + 1$ in $\catB$. Abstractly, $\Partial(\catB)$ is described as the \emph{Kleisli category} of the \emph{lift monad} $(-) + 1$ on $\catB$. In the case $\catB = \OpCat(\Theta)$, $\Partial(\catB)$ is indeed isomorphic to $\ParOpCat(\Theta)$, as indicated above. 

Now certainly any \operation{} of $\Theta$ in particular forms a partial \operation{}. Categorically, there is an identity-on-objects functor $\totaspar{-} \colon \cat B \to \Partial(\catB)$ given by:
\begin{equation*} 
\big(
\begin{tikzcd}
A \tikzcdpartialright{ \totaspar{f} } & B 
\end{tikzcd}
\big)
=
\big(
\begin{tikzcd}
 A \rar{f} & B \rar{\coproj_1} &  B + 1
\end{tikzcd}
\big)
\end{equation*}

\noindent The category $\Partial(\catB)$ inherits nice properties from $\catB$ in general:

\begin{itemize}
\item the terminal object $1$ from $\catB$ provides a distinguished object $I$ of $\Partial(\catB)$ and family of arrows $\discard{A} \colon A \partialto I$ given by $\coproj_1 \circ \unique \colon A \to 1 + 1$ in $\catB$. 
\item the initial object $0$ of $\catB$ forms a zero object in $\Partial(\catB)$, with each zero arrow $0_{A,B} \colon A \partialto B $ given by the arrow $\coproj_2 \circ \unique \colon A \to B + 1$ of $\catB$. 
\item  for any pair of objects $A$, $B$ of $\catB$, the coproduct $A + B$ in $\catB$ is again a coproduct in $\Partial(\catB)$, with coprojections $\totaspar{\coproj_1} \colon A \partialto A + B$ and $\totaspar{\coproj_2} \colon B \partialto A + B$. Hence $\Partial(\catB)$ has finite coproducts also. 
\item  when $\catB$ is symmetric monoidal with the $\otimes$ distributing over coproducts, so is $\Partial(\catB)$. The tensor $A \otimes B$ on objects is the same as in $\catB$, and satisfies $\totaspar{f} \otimes \totaspar{g} = \totaspar{f \otimes g}$, for all $f$, $g$ in $\catB$, with the coherence isomorphisms all coming from $\catB$. 
\end{itemize}  

We can at last understand the condition~\ref{opcatjointmonic}) in the definition of an operational category $\catB$: it simply asserts the joint monicity of the maps $\triangleright_i: A + A \partialto A$ in the category $\Partial(\catB)$, and  so corresponds to the fact that (partial) \operations{} in $\Theta$ are determined by their individual events. In fact, the other conditions of Definition~\ref{def:opcatpartial} also follow:

\begin{theorem} \label{thm:opcattopartialopcat}
Let $\catB$ be a (monoidal) operational category. Then $\Partial(\catB)$ is a (monoidal) operational category in partial form. 
\end{theorem}
\begin{proof}
As outlined above. In particular, condition~\ref{paropcatuniquetotal}) follows from the pullback defining an operational category. 
\end{proof}

Conversely, we have seen that $\OpCat(\Theta)$ sits inside $\ParOpCat(\Theta)$ as its subcategory of total morphisms. More generally, we have the following.

\begin{theorem} \label{thm:paropcatToOpCat}
Let $\catC$ be a (monoidal) operational category in partial form. Then $\catC_{\total}$ is a (monoidal) operational category. 
\end{theorem}
\begin{proof}
By construction, $\catC_{\total}$ has $I$ as a terminal object. Coproducts in $\catC$ restrict to $\catC_{\total}$, since the coprojections $\coproj_i$ are total by condition~\ref{paropcatatdiscardConds}). 
In the monoidal case, thanks to condition~\ref{monoidalparopcatdiscard}), the symmetric monoidal structure on $\catC$ restricts to $\catC_{\total}$ as expected, with all of the coherence and distributivity isomorphisms being total just as in Lemma~\ref{lemma:OTCConsequences}. The remaining axioms are straightforward to verify.
\end{proof}

Finally, we note that passing back and forth between these notions is indeed an equivalence.

\begin{theorem} \label{thm:opparopcatEquivalence} Let $\catB$, $\catC$ be (monoidal) operational categories in total and partial form, respectively. Then there are (monoidal) isomorphisms $\catB \simeq \Partial(\catB)_\total$ and $\catC \simeq \Partial{(\catC_{\total})}$. 
\end{theorem}
\begin{proof}
The (symmetric monoidal) functor $\totaspar{-} \colon \catB \to \Partial(\catB)_{\total}$ is full and faithful by the defining pullback of operational categories. Conversely, by condition \ref{paropcatuniquetotal}) of Definition~\ref{def:opcatpartial} the (symmetric monoidal) functor $\triangleright \colon \Partial{(\catC_{\total})} \to \catC$ sending each $(f \colon A \to B + I)$ to $(\triangleright_1 \circ f \colon A \to B)$ is also full and faithful. 
\end{proof}

\section{From Operational Categories to Theories} \label{sec:OpCatsToOpTheories}

Let us now make clear how one can use operational categories to study an OTC $\Theta$. We mainly focus on the operational category in partial form $\catC = \ParOpCat(\Theta) \simeq \Partial(\OpCat(\Theta))$, with arrows and composition denoted $\catpopsto$, $\catpopscomp$. For any system $A$ of $\Theta$, let us now simply write $A$ for the corresponding object $\indsys{A}$ of $\catC$, so that a general object $\indsys{A_x}_{x \in X}$ is given by the coproduct $\coprod_{x \in X} A_x$ in $\catC$. 

\paragraph{Events, states and effects} 
Firstly, each event $f \colon A \eventto B$ in $\Theta$ has a corresponding arrow $f \colon A \partialto B$ in $\catC$, including the impossible events $0 \colon A \partialto B$. In particular, states, effects and scalars are given by arrows $\omega \colon I \partialto A$, $e \colon A \partialto I$ and $r \colon I \partialto I$ in $\catC$ respectively. Each effect then also corresponds to an arrow $e: A \catopsto 1 + 1$ in the total category $\catB = \OpCat(\Theta)$, with complementary effect given by $e^{\bot} = [\coproj_2, \coproj_1] \circ e: A \catopsto 1 + 1$.

\paragraph{\Operations{}}
As we've seen, a collection of events $\{f_x \colon A \partialto B_x\}_{x \in X}$ forms a partial \operation{} iff there is an arrow $f \colon A \partialto \coprod_{x \in X} B_x$ in $\catC$ with $\triangleright_x \partialcomp f = f_x$ for all $x \in X$. This \operation{} is total iff $f$ is total in the sense that $\discard{} \partialcomp f = \discard{}$. Conversely, thanks to the joint monicity of the projections $\triangleright_x$, each arrow $f \colon A \partialto \coprod_{x \in X} B_x$ is determined by its collection of events $\{\triangleright_x \partialcomp f \colon A \partialto{} B_x \}_{x \in X}$. In the monoidal case, (partial) \operations{} may be composed spatially $f \otimes g$, as one would expect. 

\paragraph{Control structure}
The ability to form controlled \operations{} is inherent in the coproduct structure of $\catC$. Given any (partial) \operation{} $\{f_i \colon A \catpopsto B_i\}^n_{i=1}$, and for each outcome a (partial) \operation{} $g_i \colon B_i \partialto C_i$, for some object $C_i = \indsys{C_j}_{j \in X_i}$ of $\catC$, the corresponding controlled \operation{} is given by the morphism:
\[
\begin{tikzcd}[column sep = large]
A \tikzcdpartialright{f} & B_1 + \ldots + B_n \tikzcdpartialright{ g_1 + \ldots + g_n } & C_1 + \ldots + C_n
\end{tikzcd}
\]

\paragraph{Coarse-graining} 
The coproducts in $\catC$ also neatly capture the coarse-graining structure of $\Theta$. For any compatible collection of events $\{f_x \colon A \eventto B \}_{x \in X}$, corresponding to some partial arrow $f \colon A \partialto \coprod_{x \in X} B$, their coarse-graining $\bigovee_{x \in X} f_x$ is given by the morphism:
\[
\begin{tikzcd}
A \tikzcdpartialright{f} & \coprod_{x \in X} B \tikzcdpartialright{\triangledown} & B
\end{tikzcd}
\]
where $\triangledown$ is the \emph{codiagonal} map, defined by $\triangledown \partialcomp \coproj_x = \id{B}$ for all $x \in X$. Intuitively, writing $n \cdot A = \coprod^n_{i = 1} A$, we think of each codiagonal $\triangledown \colon n \cdot A \partialto A$ as `deleting' the classical information stored in the coproduct $n \cdot A$. Indeed, whenever $\Theta$ is monoidal, we may see $\triangledown$ as discarding the $n$-bit classical system $n \eqdef n \cdot I$, since we have: 
\vspace{-5pt}
\[
\big(
\begin{tikzcd}
A \otimes n \simeq n \cdot A \tikzcdpartialright{\triangledown} &  A
\end{tikzcd}
\big)
\qquad
=
\qquad
\big( 
\xymatrix{
A \otimes n \ar[r]^-{ \hspace{3pt} \id{} \otimes \hspace{1.4pt} \discard{} \ } & A \otimes I \simeq A
}
\big)
\qquad
=
\qquad
\begin{pic}[scale=0.8]
\node[ground] (g1) at (0.5,0){};
\node[below] (b) at (0,-0.5){A};
\node[above] (c) at (0,0.5){A};
\draw (g1.south) to (0.5,-0.52) node[below]{$n$}; 
\draw (b.north) to (c.south);
\end{pic}
\]
\vspace{-10pt}



\noindent \textbf{Convex structure}. Since we may express all of the basic notions of an OTC categorically, so may we any derived ones, such as (sub)convex combinations of states: 
\[
 \bigovee^n_{i=1} r_i \bullet \omega_i
 =
 \big(
\begin{tikzcd}[column sep = large]
I \tikzcdpartialright{r} & n \cdot I \tikzcdpartialright{[\omega_1, \ldots , \omega_n]} & A
\end{tikzcd}
\big)
\]
or more general events in a monoidal theory: 
\[
\bigovee^n_{i =1} r_i \bullet f_i =
\big(
\begin{tikzcd}[column sep = large]
A \simeq A \otimes I \tikzcdpartialright{\id{} \otimes r} & A \otimes n \simeq n \cdot A \tikzcdpartialright{[f_1, \ldots ,f_n]} & B 
\end{tikzcd}
\big)
\]

\subsection{Defining an OTC from an operational category} \label{sec:OTCFromOPCat}

The above ideas suggest another way of looking at any operational category (in partial form) $\catC$. 
Since the notions of \operation{} and coarse-graining make sense for arbitrary arrows $f \colon A \to B$ in $\catC$, rather than seeing them as partial \operations{} in some OTC $\Theta$, we can alternatively view them as the events of a new OTC extending $\Theta$, in the following way. 

\begin{theorem} \label{thm:ParFormToOTC} Any (monoidal) operational category in partial form $(\catC, I, \discard{})$  defines a (monoidal) OTC, denoted $\OTCFrom{\catC}$, with $\catC$ as its category of events and trivial system $I$, as follows:
\begin{itemize} 
\item  a finite collection of events $\{f \colon A \to B_x\}_{x \in X}$ forms a partial \operation{} iff there exists some $f \colon A \to \coprod_{x \in X} B_x$ in $\catC$ with $\triangleright_x \circ f = f_x$ for all $x \in X$, and this partial \operation{} is total iff $f$ is total in $\catC$. 
\item the coarse-graining of a pair of compatible events $f$, $g \colon A \partialto{} B$ is given by $f \ovee g  = \triangledown \circ h$, where $h \colon A \to B + B$ is the unique arrow with $\triangleright_1 \circ h = f$ and $\triangleright_2 \circ h = g$.
\end{itemize}
\end{theorem}

\begin{proof}
We know that for each event $f \colon A \to B$ in $\catC$, there is a (unique) total $g \colon A \to{} B + I$ with $f = \triangleright_1 \circ g$, and so $f$ belongs to some \operation{}, namely $\{f, \triangleright_2 \partialcomp g\}$. 
Complementary effects $e^{\bot}$ are given as above. The control structure on \operations{} also comes from the coproducts of $\catC$ as above, along with the fact that $\discard{A + B} = [\discard{A}, \discard{B}] \colon A + B \to I$. It's straightforward to check that each of the coarse-graining equations are satisfied, with the zero arrows $A \to 0 \to B$ behaving as the impossible events. In particular, in the monoidal case, the law $f \otimes (g \ovee h) = (f \otimes g) \ovee (f \otimes h)$ follows from distributivity of the $\otimes$ over the coproducts. Similarly, \operations{} are preserved by $\otimes$ since the distributivity isomorphisms are total. 
\end{proof}

Hence we may alternatively view any operational category in partial form $\catC$ as the category $\Events_{\Theta}$ of events of an OTC $\Theta = \OTCFrom{\catC}$, with $\catB = \catC_{\total}$ then forming its subcategory $\DetEvents_{\Theta}$ of deterministic events. 

\begin{keyexample*}\label{example:OTCPLus} For any OTC $\Theta$, we define a new OTC $\Theta^{\otcplus}$ by setting $\Theta^{\otcplus} \eqdef \OTCFrom{\ParOpCat(\Theta)}$. Explicitly, as before, systems in $\Theta^+$ are finite indexed collections $\indsys{A_x}_{x \in X}$ of systems of $\Theta$, with events $M  \colon \indsys{A_x}_{x \in X} \eventto \indsys{B_y}_{y \in Y}$ given by matrices of events from $\Theta$ in which each column forms a partial \operation{}. Each partial \operation{} $\{f_x \colon A \eventto B_x\}_{x \in X}$ in $\Theta$ then corresponds to a single event $f \colon \indsys{A} \eventto \indsys{B_x}_{x \in X}$ in $\Theta^{\otcplus}$. 
\end{keyexample*}

\subsection{Direct sum systems} \label{subsec:Direct Sums}
Now, for each operational category in partial form $\catC$, the theory $\OTCFrom{\catC}$ comes with a useful extra property. The coproducts in $\catC$ provide the ability to `add systems together', which we characterise operationally in the following way. 

\begin{definition}
In any OTC $\Theta$, an indexed collection $\{A_x\}_{x \in X}$ of systems has a \emph{direct sum} if there is a system $\bigoplus_{x \in X} A_x$ and \operation{} $\{ \triangleright_y \colon \bigoplus_{x \in X} A_x \eventto A_y\}_{y \in X}$ such that, for each partial \operation{} 
$\{ f_x \colon B \eventto A_x\}_{x \in X}$ there is a unique event $f  \colon B \eventto \bigoplus_{x \in X} A_x$ satisfying $\triangleright_x \eventcomp f = f_x$ for all $x \in X$. We say $\Theta$ \emph{has direct sums} if each finite such collection has a direct sum.
\end{definition}

Our terminology goes back to~\cite{chiribella2010purification}, where the use of direct sum systems is proposed in the context of operational-probabilistic theories. The presence of direct sums allows us to consider (partial) \operations{} $\{f_x \colon A \eventto B_x \}_{x \in X}$ more simply as single events $f \colon A \eventto \bigoplus_{x \in X} B_x$, just as the coproducts in an operational category do. In fact, both concepts are equivalent.


\begin{lemma} \label{OTCPlus_Coproduct_iff} 
For a non-empty finite collection $\{A_x\}_{x \in X}$ of systems, and further system $A$, the following are equivalent:  
\begin{enumerate}[i)]
\item $A$ forms a direct sum $(A = \bigoplus_{x \in X} A_x, \triangleright_x)$; 
\item There is a \operation{} $\{\triangleright_x \colon A \eventto A_x\}_{x \in X}$ and collection of events $\coproj_x \colon A_x \eventto A$ satisfying the equations \eqref{eq:projectionarrows_partial} and $\bigovee_{x \in X} \coproj_x \eventcomp \triangleright_x = \id{A}$;

\item $A$ forms a coproduct $(A = \coprod_{x \in X} A_x, \coproj_x)$ in $\Events_{\Theta}$, for which the events $\triangleright_x \colon A\eventto A_x$ defined by \eqref{eq:projectionarrows_partial} are jointly monic, and form \anoperation{} $\{\triangleright_x\}_{x \in X}$. 
\end{enumerate}
\end{lemma}
\begin{proof}

(i) $\Rightarrow$ (ii): 
We define $\coproj_x \colon A_x \eventto A$ to be the unique event corresponding to the \operation{} 
$\{ \id{A_x} \colon A_x \eventto A_x\} \cup \{0 \colon A_y \eventto A_x \}_{y \neq x}$. Then, using control, the event $\bigovee_{x \in X} \coproj_x \eventcomp \triangleright_x$ is well-defined, and satisfies: 
\[
\triangleright_y \eventcomp (\bigovee_{x \in X} \coproj_x \eventcomp \triangleright_x) = 
\bigovee_{x \in X} (\triangleright_y \eventcomp \coproj_x \eventcomp \triangleright_x )=
\triangleright_x
\]
and so by uniqueness is equal to $\id{A}$. 

(ii) $\Rightarrow$ (iii): 
For any collection of events $f_x \colon A_x \eventto B$, if $f \colon A \eventto B$ satisfies $f \eventcomp \coproj_x = f_x$ for all $x \in X$ then $f = f \eventcomp (\bigovee_{x \in X} \coproj_x \eventcomp \triangleright_x) = \bigovee_{x \in X} f_x \eventcomp \triangleright_x$. Hence this defines the unique such $f$.

(iii) $\Rightarrow$ (i): For any partial \operation{} $\{f_x \colon B \eventto A_x\}_{x \in X}$, the event 
$f = \bigovee_{x \in X}(\coproj_x \eventcomp f_x) \colon B \eventto A$ is well-defined and
satisfies  $\triangleright_x \eventcomp f = f_x$ for all $x \in X$. It is unique by joint monicity of the $\triangleright_x$. \end{proof}

Further, unravelling the definitions gives that an empty direct sum is the same as a terminal object in $\Event_{\Theta}$, which is then a zero object thanks to the family of events $0_{A,B} \colon A \eventto B$.

\begin{lemma} \label{OTCWithDirSums_distrib}
For any monoidal OTC $\Theta$ with direct sums, in $\Events_{\Theta}$ the tensor $\otimes$ distributes over the coproducts, and hence the direct sums. 
\end{lemma}
\begin{proof} 
Using Assumptions~\ref{assump_composition} and \ref{assump:control}, along with the coarse-graining equations, one may verify that the event 
$\coproj_1 \eventcomp (\id{A} \otimes \triangleright_B) \ovee \coproj_2 \eventcomp (\id{A} \otimes \triangleright_C) \colon A \otimes (B + C) \eventto A \otimes B + A \otimes C$ 
is well-defined and inverse to $[\id{A} \otimes \coproj_1, \id{A} \otimes \coproj_2]$. We also have $0 \otimes A \simeq 0$ since $\id{0 \otimes A} = 0 \otimes \id{A} = 0$. 
\end{proof}

\subsection{Equivalence of operational categories and theories with direct sums}

Thanks to the above characterisation of direct sums in terms of coproducts, we have:  

\begin{corollary} \label{cor:OTCFromHasDSums}
For any (monoidal) operational category in partial form $\catC$, the (monoidal) theory $\OTCFrom{\catC}$ has direct sums.
\end{corollary}

In particular, starting from any OTC $\Theta$ we may always pass to the extended one $\Theta^{\otcplus} = \OTCFrom{\ParOpCat(\Theta)}$ with direct sums, without altering $\Theta$ if they were already present: 

\begin{theorem} \label{thm:EquivIffDirSums} For every OTC $\Theta$, the theory $\Theta^{\otcplus}$ has direct sums. Conversely, $\Theta$ has direct sums iff there is an equivalence of (monoidal) theories $\Theta \simeq \Theta^{\otcplus}$, preserving direct sums. 
\end{theorem}
\begin{proof}
$\Theta^+$ has direct sums by Corollary~\ref{cor:OTCFromHasDSums}. Hence if $\Theta$ and $\Theta^+$ are equivalent then $\Theta$ must also. 
Conversely, if $\Theta$ has direct sums, consider the assignment which sends each system $\indsys{A_x}_{x \in X}$ of $\Theta^+$ to the system $\bigoplus_{x \in X} A_x$ of $\Theta$, and each event $M \colon \indsys{A_x}_{x \in X} \eventto \indsys{B_y}_{y \in Y}$ to the unique event $\hat{M} \colon \bigoplus_{x \in X} A_x \eventto \bigoplus_{y \in Y} B_y$ satisfying $\triangleright_y \eventcomp \hat{M} \eventcomp \coproj_x = M(x,y) \colon A_x \eventto B_y$, using Lemma~\ref{OTCPlus_Coproduct_iff}.
It's straightforward to check that this defines a (monoidal) equivalence of categories preserving discarding $\discard{}$ and the coproducts, using Lemma~\ref{OTCWithDirSums_distrib} in the monoidal case. By our next result, this in fact ensures that $\Theta$ and $\Theta^{\otcplus}$ are equivalent as theories.
\end{proof}

Now we've seen that once direct sums are present, they can be described equivalently as coproducts. In fact, these coproducts encode the full structure of the theory, just as in the definition of the theory $\OTCFrom{\catC}$. As we've seen, any partial \operation{} $\{f_x \colon A \eventto B_x\}_{x \in X}$ is described by a single event $f \colon A \eventto \bigoplus_{x \in X} B_x \simeq \coprod_{x \in X} B_x$, and will be total precisely when $f$ is deterministic. Further, coarse-graining may again be described using the codiagonal maps, since we have:
\begin{equation*}
\triangledown \eventcomp f  = \triangledown \eventcomp (\bigovee_{x \in X} \coproj_x \eventcomp \triangleright_x) \eventcomp f
 = \bigovee_{x \in X} (\triangledown \eventcomp \coproj_x) \eventcomp (\triangleright_x \eventcomp f) 
 = \bigovee_{x \in X} (\id{}) \eventcomp f_x = \bigovee_{x \in X} f_x
\end{equation*}
for each compatible collection $\{f_x \colon A \eventto B\}_{x \in X}$, with corresponding event $f \colon A \eventto \coprod_{x \in X} B_x$. We have established the following.

\begin{lemma} \label{thm:OTCDSumstoopcat}
For any OTC $\Theta$ with direct sums, $\catC = \Events_{\Theta}$ is an operational category in partial form, with $\Theta = \OTCFrom{\catC}$. 
\end{lemma}
\begin{proof}
By Theorem~\ref{thm:EquivIffDirSums}, we have a (monoidal) equivalence of categories $\Events_{\Theta} \simeq \Events_{\Theta^{\otcplus}} \simeq \ParOpCat(\Theta)$ preserving coproducts and discarding, and so $\catC$ indeed forms an operational category in partial form. As outlined above, \operations{} and coarse-graining are defined in $\Theta$ just as in $\OTCFrom{\catC}$. 
\end{proof}

We have reached our first main result, summarising Corollary~\ref{cor:OTCFromHasDSums}, Lemma~\ref{thm:OTCDSumstoopcat} and Section~\ref{subsec:totparequiv}.

\begin{theorem} \label{thm:MainEquivalence}
The following structures are equivalent:
\begin{itemize} 
\item a (monoidal) operational theory with control $\Theta$ with direct sums;
\item a (monoidal) operational category in partial form $\catC$;
\item a (monoidal) operational category $\catB$; 
\end{itemize}
under the correspondences $\catC = \Events_{\Theta}$, $\Theta = \OTCFrom{\catC}$, $\catB \simeq \catC_{\total}$ and $\catC \simeq \Partial(\catB)$.
\end{theorem}

\subsection{Examples}
Now that we understand the relationship between operational categories and theories, let's briefly look again at each of our main examples of OTCs. For further details on these categories, see \cite{EffectusIntro}. Most of our example theories $\Theta$ in fact have direct sums, and so are determined by their categories of events $\Events_{\Theta} \simeq \ParOpCat(\Theta)$, or deterministic events $\DetEvents_{\Theta} \simeq \OpCat(\Theta)$, which are operational categories in partial and total form, respectively.

\begin{enumerate}[i)]
\item The category of events of $\ClassDet$ is the category $\cat{PFun}$ of sets and partial functions, with direct sums given by disjoint union of sets. The deterministic events form the operational category $\Set$ of sets and (total) functions. 

\item The theory $\ClassProb$ has direct sums described in the same way. An event $f \colon A \eventto B$ here is deterministic when it sends each $a \in A$ to a (normalised) distribution over $B$. Their category is described abstractly as the \emph{Kleisli category} of the \emph{distribution monad}, $\Kleisli{\mathcal{D}}$.

\item In contrast, the theory $\FinHilbOTC$ does not have direct sums. Its direct sum completion $\FinHilbOTC^{\otcplus}$ 
is the theory $\FinCStarOTC$ of finite-dimensional C*-algebras, via the correspondence $\indsys{\HilbH_x}_{x \in X} \mapsto \bigoplus_{x \in X} \curlyB(\HilbH_x)$. Direct sums also exist in the infinite-dimensional case $\CStarOTC$. In both cases, the corresponding `total' operational category is the category $\FDimCStarBracket$, of (finite-dimensional) C*-algebras and unital, completely positive maps, while its partial form $\FDimCStarSUBracket$ instead has as arrows completely positive, sub-unital maps. 

\item Finally, $\Mat_R$ has direct sums, given on systems $n \in \mathbb{N}$ by addition of natural numbers.

\end{enumerate}

\section{Further Operational Assumptions} \label{sec:AdditionalAssumptions}

Our definition of an OTC was deliberately chosen to be as weak as possible while still allowing for the categorical approach presented above. There are further basic requirements that one might expect to form a part of our framework, such as the following. 

\begin{axiom}[Positivity] \label{axiom:positivity} Whenever $\{f_x\}_{x \in X}$ and $\{f_x \}_{x \in X} \cup \{g_y\}_{y \in Y}$ both form \operations{}, we have $g_y = 0$ for all $y \in Y$.
\end{axiom}

Intuitively, since on any run of the first \operation{} one of the events $f_x$ must occur, each of the events $g_y$ must be impossible. This condition translates categorically as follows. 

\begin{lemma} For any OTC $\Theta$, the following are equivalent:
\begin{enumerate}[i)]
\item $\Theta$ is positive;
\item Events in $\Theta$ satisfy $f \ovee g = 0 \implies f = g = 0$ and $\discard{} \eventcomp f = 0 \implies f = 0$;
\item Events in $\Theta^{\otcplus}$ satisfy $\discard{} \eventcomp f = 0 \implies f = 0$;
\item In $\OpCat(\Theta)$, diagrams of the following form are pullbacks: 
\begin{equation} \label{eq:posOpCatPullback}
\begin{tikzcd}
A \rar{\unique} \arrow[d, swap, "\coproj_1"] & 1 \dar{\coproj_1} \\
A + B \rar[swap]{\unique + \unique} & 1 + 1
\end{tikzcd}
\end{equation}
\end{enumerate}

We will call any operational category $\catB$ with this property \emph{positive}. Note that the pullback in the definition of an operational category is a special case of \eqref{eq:posOpCatPullback}. 
\end{lemma}
\begin{proof} 
When interpreted in $\OpCat(\Theta)$, the above pullback states that for any \operation{} $\{f_x \colon C \eventto A_x \}_{x \in X} \cup \{g_y \colon C \eventto B_y \}_{y \in Y}$ in $\Theta$ satisfying 
\begin{equation} \label{eq:poscatproof}
\bigovee_{y \in Y} \discard{} \eventcomp g_y = 0
\end{equation}
we have $g_y = 0$ for all $y \in Y$. Equivalently, any partial \operation{} $\{g_y \colon C \eventto B_y\}_{y \in Y}$ satisfying \eqref{eq:poscatproof} has $g_y = 0$ for all $y \in Y$. This gives (iii) by the definition of $\Theta^{\otcplus}$, and is easily seen to be equivalent to each of (i) and (ii).
\end{proof}

\begin{examples*} \label{examples:PosOpCats} Each of our leading examples of operational theories $\ClassDet$, $\ClassProb$, $\FinHilbOTC$ and $\mathsf{(Fin)CStar}$ are positive, and hence so are their corresponding operational categories $\Set$, $\Kl(\mathcal{D})$ and $\FDimCStarBracket$. The theory $\Mat_{\mathbb{Z}}$ is not positive, since it comes with non-zero scalars $1$ and $-1$ satisfying $1 \ovee -1 = 0$.
\end{examples*}

The positivity axiom comes with a few nice consequences, which are discussed in Appendix~\ref{appendix:propertiesOfOpCats}. For example, in the category $\catC = \Partial(\catB)$ the discarding maps $\discard{A} \colon A \partialto I$ are now uniquely determined, rather than having to be stated as extra structure. Moreover, isomorphisms in $\Partial(\catB)$ are always total, i.e.\ deterministic, as one would expect when viewing them as reversible physical events. Categorically, the initial object $0$ becomes \emph{strict}, and \eqref{eq:posOpCatPullback} extends to more general pullbacks:
\begin{equation}
\label{eq:generalPosPullbacks}
\begin{tikzcd}
A \rar{f} \dar[swap]{\coproj_1} & B \dar{\coproj_1} \\
A + C \rar[swap]{f + g} & B + D
\end{tikzcd}
\end{equation}

Beyond positivity, there are stronger requirements one might wish to adopt on purely operational grounds, such as rules ensuring the scalars $r: I \eventto I$ behave even more like probabilities. The strongest assumption we can make of more general events is to identify those which are `testably the same', as follows. 

\paragraph{Operational Equivalence}
We say two events $f$, $g \colon A \eventto B$ of a monoidal OTC $\Theta$ are \emph{operationally equivalent}, and write $f \monopequiv g$, when   
\[ \label{firstopquotient}
 \begin{pic}
  \node[effectwide](e) at (1,1.5){$e$};
  \node[statewide](w) at (1,0){};
  \draw ([xshift=3mm]w.north west |-,0.75) node[mor](f) {$f$};
  \node[scale = 1.6] at (1,0){$\omega$};
  \draw([xshift=3mm] w.north west) to (f.south);
  \draw([xshift=-3mm] w.north east) to ([xshift=-3mm] e.south east);
  \draw (f.north) to ([xshift=3mm] e.south west);
  \end{pic}
  =
  \begin{pic}
  \node[effectwide](e) at (1,1.5){$e$};
  \node[statewide](w) at (1,0){};
  \draw ([xshift=3mm]w.north west |-,0.75) node[mor](f) {$g$};
  \node[scale = 1.6] at (1,0){$\omega$};
  \draw([xshift=3mm] w.north west) to (f.south);
  \draw([xshift=-3mm] w.north east) to ([xshift=-3mm] e.south east);
  \draw (f.north) to ([xshift=3mm] e.south west);
  \end{pic}
\]  


\noindent for all external systems $C$, states $\omega \colon I \eventto A \otimes C$, and effects $e \colon B \otimes C \eventto I$.  
In a non-monoidal theory, we instead consider the simpler condition that $f \simeq g$ whenever $e \eventcomp f \eventcomp \omega = e \eventcomp g \eventcomp \omega$ for all states $\omega \colon A \eventto I$ and effects $e \colon B \eventto I$. Both define equivalence relations, 
intuitively with $f \monoropequiv g$ whenever $f$ and $g$ give the same probabilities to all possible experiments. 
We call a (monoidal) OTC \emph{(monoidally) separated} if $f \monoropequiv g \implies f = g$ for all events $f$, $g$. In fact, all of our main examples of OTCs are separated. 
The significance of separation is discussed in \cite{EPTCS172.1}, on which the following result is based.

\paragraph{The Quotient OTC} Given any (monoidal) OTC $\Theta$, we define a new (monoidally) separated OTC $\Theta/{\monoropequiv}$ as follows. 
Events $f \colon A \eventto B$ in $\Theta/{\monoropequiv}$ are equivalence classes $[g]$ of events $g \colon A \eventto B$ of $\Theta$ under $\monoropequiv$. \Operations{} are collections $\{[f_x] \colon A \eventto B_x \}_{x \in X}$ for which there is some \operation{} $\{ f_x \colon A \eventto B_x \}_{x \in X}$ in $\Theta$, with coarse-graining defined by $[f]\ovee [g] = [f \ovee g]$.

\begin{theorem} \label{thm:OTCSeparation}  
Give any (monoidal) OTC $\Theta$,  $\Theta/{\monoropequiv}$ is a well-defined (monoidal) OTC which is (monoidally) separated. If $\Theta$ is separated, then $\Theta$ and $\Theta/{\monoropequiv}$ are isomorphic theories. 
Further, when $\Theta$ has direct sums, so does $\Theta/{\monoropequiv}$.
\end{theorem}
\begin{proof}
Assumption~\ref{assump:complementaryeffects} is the only interesting one to check, requiring us to show that for all effects $a, b \colon A \eventto I$ with $a \monoropequiv b$ we have $a^{\bot} \monoropequiv b^{\bot}$.  But for all states $\omega \colon I \eventto A$, we have:
\[
\discard{A} \eventcomp \omega = (a \ovee a^{\bot}) \eventcomp \omega = a \eventcomp \omega \ovee a^{\bot} \eventcomp \omega = b \eventcomp \omega \ovee a^{\bot} \eventcomp \omega = b \eventcomp \omega \ovee b^{\bot} \eventcomp \omega
\]
and so, thanks to cancellativity of scalars in $\Theta$, $a^{\bot} \eventcomp \omega = b^{\bot} \eventcomp \omega$, as required. Clearly $\Theta / \monoropequiv$ is (monoidally) separated, and if $\Theta$ is then both theories are identical. 
Finally, if $\Theta$ has direct sums, then the events $[\triangleright_x] \colon \bigoplus_{y \in X} A_y \eventto A_x$, for $x \in X$, form a direct sum in $\Theta/{\monoropequiv}$.
\end{proof}

In particular, starting from any (monoidal) operational category $\catB$, we may always pass to a new one $\catB/{\monoropequiv}$ for which the theory $\OTCFrom{\Partial(\catB)}$ is (monoidally) separated, in the same way.


\section{Effectus Theory} \label{sec:EffectusTheory}

The categorical structures we have made use of above were first considered by Jacobs et al.\ \cite{NewDirections2014aJacobs,StatesConv2015JWW,Partial2015Cho} in a recent approach to the study of quantum computation using categorical logic called \emph{effectus theory}. An accessible introduction to effectus theory is given in \cite{EffectusIntro}. 

\begin{definition}
\cite{EffectusIntro}
A \emph{(monoidal) effectus} is a positive (monoidal) operational category for which diagrams of the following form are pullbacks:
\begin{equation}  \label{eq:effectusdefnpullback}
\begin{tikzcd}
A + B \rar{\unique + \id{}} \dar[swap]{\id{} + \unique} & 1 + B \dar{\unique + \unique} \\
A + 1 \rar[swap]{\unique + \unique} & 1 + 1
\end{tikzcd}
\end{equation}
\end{definition}

Thanks to our results we may now give an operational interpretation to effectus theory, equating effectuses with certain OTCs. We will consider theories satisfying the following axiom.

\begin{axiom}[Observations determine \operations{}] \label{axiom:obsdetops}
A collection of events $\{f_x \colon A \eventto B_x\}_{x \in X}$ forms \anoperation{} in $\Theta$ whenever $\{\discard{} \eventcomp f_x \colon A \eventto I \}_{x \in X}$ does.
\end{axiom}

When combined with the ability to form direct sums, this implies the following.

\begin{axiom}[Combining] \label{axiom:combiningops}
Any pair of partial \operations{} $\{f_x \colon A \eventto B_x\}_{x \in X}$, $\{g_y \colon A \eventto C_y\}_{y \in Y}$ satisfying $\bigovee_{x \in X} \discard{} \eventcomp f_x = (\bigovee_{y \in Y} \discard{} \eventcomp g_y)^{\bot}$
form \anoperation{} $\{f_x\}_{x \in X} \cup \{g_y\}_{y \in Y}$.
\end{axiom}

\begin{lemma} \label{lemma:EffAsOpCat} The following are equivalent for an OTC $\Theta$:
\begin{enumerate}[i)]
\item Diagrams of the form \eqref{eq:effectusdefnpullback} are pullbacks in $\OpCat(\Theta)$; 
\item Observations determine \operations{} in $\Theta^{\otcplus}$;
\item $\Theta$ satisfies combining.
\end{enumerate}
\end{lemma}

\begin{proof}
(i) $\iff$ (iii): For any operational category $\catB$, the pullback \eqref{eq:effectusdefnpullback} states precisely that  for any two partial arrows $f \colon C \partialto A$, $g \colon C \partialto B$ with  $\discard{} \partialcomp f = (\discard{} \partialcomp g)^{\bot}$ there is some (necessarily total) $h \colon C \partialto A + B$ with $\triangleright_1 \partialcomp h = f$, $\triangleright_2 \partialcomp h = g$. In the case $\catB = \OpCat(\Theta)$, this is precisely the condition (iii).

(ii) $\Rightarrow$ (iii): Suppose that partial \operations{} $f \colon C \catpopsto A$, $g \colon C \catpopsto B$, satisfy $\discard{} \catpopscomp f = (\discard{} \catpopscomp g)^{\bot}$. Then $\{\discard{} \catpopscomp f, \discard{} \catpopscomp g \}$ is \anoperation{} in $\Theta^+$, and since observations determine \operations{}, so is $\{f, g\}$, i.e.\ there exists an $h \colon C \catpopsto A + B$ as above.

(iii) $\Rightarrow$ (ii): As a special case of (iii), we have that whenever $\{\discard{} \circ f_1\} \cup \{f_2, \ldots, f_n \}$ forms \anoperation{} in $\Theta$, then so does $\{f_i\}_{i= 1}^n$. From this, it follows inductively that observations determine \operations{} in $\Theta$. Since $\OpCat(\Theta^{\otcplus}) \simeq \OpCat(\Theta)$, this shows that they also do in $\Theta^{\otcplus}$. 
\end{proof}

This gives our next main result, which provides a new operational understanding of the effectus axioms.

\begin{corollary} \label{cor:Effiscetainotcs} The following structures are equivalent:
\begin{itemize}
\item a (monoidal) effectus $\catB$; 
\item a positive (monoidal) OTC with direct sums $\Theta$, in which observations determine \operations{};
\end{itemize}
under the correspondences $\catB \simeq \OpCat(\Theta) \simeq \DetEvents_{\Theta}$, $\Theta \simeq \OTCFrom{\Partial(\catB)}$.
\end{corollary}

\begin{examples*}
Our main examples $\ClassDet$, $\ClassProb$, $\FinHilbOTC$ and $\mathsf{(Fin)CStar}$ all have observations determining their \operations{}, and hence their categories of \operations{} $\Set$, $\Kleisli{\mathcal{D}}$ and $\FDimCStarBracket$ are all monoidal effectuses. 
\end{examples*}

The mild extra physical assumption that observations determine \operations{} has several pleasing consequences for the operational category $\catB$, which are explored in \cite{NewDirections2014aJacobs,Partial2015Cho,EffectusIntro}.
Crucially, the coarse-graining $\ovee$ now makes each makes each homset $\Partial(\catB)(A,B)$ into a \emph{partial commutative monoid}, with each space of effects in particular forming an \emph{effect algebra}, a well-known structure from quantum logic. Further, composition makes the scalars $M = \Partial(\catB)(1, 1)$ into an \emph{effect monoid}, and our earlier description of `convex combinations' of arrows in $\Partial(\catB)$ is then made precise by the result that these homsets form \emph{(sub)convex sets} over the effect monoid $M$.


\section{Functorial Correspondence between Operational Theories and Categories } \label{sec:CategoricalResults}

Here, we briefly discuss how each of our main results can be stated structurally, in terms of functors between categories. Proofs for this section can be found in Appendix~\ref{appendix:omittedproofs}. 
For this, we will now need to consider morphisms \emph{between} operational theories.


\begin{definition} \label{defmorphismOTC} A \emph{morphism of OTCs} $F \colon \Theta \to \Theta'$ is a functor $F \colon \Events_{\Theta} \to \Events_{\Theta'}$ which preserves \operations{}, coarse-graining and the trivial system, in that:
\begin{itemize}
\item whenever $\{f_x \colon A \eventto B_x\}_{x \in X}$ forms \anoperation{}, so does $\{F(f_x) \colon F(A) \eventto F(B_x) \}_{x \in X}$;
\item $F(f \ovee g) = F(f) \ovee F(g)$ for all compatible events $f,g$;
\item $\discard{F(I)} \colon F(I) \eventto I$ is an isomorphism.
\end{itemize}
A \emph{morphism of monoidal OTCs} is additionally strong symmetric monoidal as a functor, with $\discard{F(I)}^{-1}$ as its coherence isomorphism $I \eventto F(I)$. We write $\MonOTC$ for the category of (monoidal) OTCs and morphisms between them.
\end{definition}

\begin{lemma} \label{lemma:morphismOTC} Any morphism of OTCs $F \colon \Theta \to \Theta'$ satisfies $\discard{F(A)} = \discard{F(I)} \eventcomp F(\discard{A}) \colon F(A) \eventto I$. Further, $F$ preserves impossible events $0_{A,B}$, and any direct sums which exist in $\Theta$.
\end{lemma}

We write $\cat{(Mon)OTC^{\otcplus}}$ for the full subcategory of $\MonOTC$ of OTCs coming with direct sums. By the above result, morphisms in $\cat{(Mon)OTC^{\otcplus}}$ preserve direct sums, as expected. 


Our first structural result shows that $\Theta^{\otcplus}$ is the OTC in which we `freely add direct sums' to an OTC $\Theta$. For this, we define $\OTCPlusStrict$ to be the category of (monoidal) OTCs $(\Theta, \oplus)$ coming with specified direct sums, and morphisms of OTCs $F \colon \Theta \to \Theta'$ which preserve them strictly, i.e.\ 
for which $F( \bigoplus_{x \in X} A_x) = \bigoplus_{x \in X} F(A_x)$ in $\Theta'$, with projection events $F(\triangleright_x)$, for each direct sum $\bigoplus_{x \in X} A_x$ in $\Theta$.

\begin{theorem} \label{thm:FreeOTCPlus} The assignment $\Theta \mapsto \Theta^{\otcplus}$ defines a left adjoint $(-)^{\otcplus}$ to the forgetful functor $U \colon \OTCPlusStrict \to \OTC$.
\end{theorem}

Next, we consider morphisms of operational categories (c.f.~\cite{Partial2015Cho},~\cite{StatesConv2015JWW}). A \emph{morphism of operational categories in partial form} $F \colon (\catC, \discard{}) \to (\catC', \discard{})$ is a functor $F \colon \catC \to \catC'$ preserving finite coproducts, and which `preserves discarding' in that
$\discard{F(I)} \colon F(I) \to I$ is an isomorphism and $ \ \discard{F(A)} = \discard{F(I)} \circ F(\discard{A}) \colon F(A) \to I$ for all objects $A$. Again, in the monoidal case, we require $F$ to be strong symmetric monoidal as a functor, with $\discard{F(I)}^{-1}$ as its coherence isomorphism $I \to F(I)$. A \emph{morphism of (monoidal) operational categories} $F \colon \catB \to \catB'$ is a (strong symmetric monoidal) functor which preserves coproducts and the terminal object. We write $\cat{(Mon)OpCat}$ and $\cat{(Mon)OpCatPar}$ for the categories of (monoidal) operational categories in total and partial form and their morphisms, respectively.

\begin{theorem} \label{thm:1catEquivalence}
The assignments of Theorem~\ref{thm:MainEquivalence} define equivalences of categories 
\begin{equation*} 
\begin{tikzcd}[bend angle = 15]
 \cat{(Mon)OTC^{\otcplus}} \arrow[rr, bend left, "\Events_{(-)}"] & { \  \simeq}&  \cat{(Mon)OpCatPar} \arrow[ll, bend left, "\OTCFrom{-}"]  \arrow[rr, bend left, "(-)_{\total}"] & {\! \simeq} & \cat{(Mon)OpCat}  \arrow[ll, bend left, "\Partial(-)"]
 \end{tikzcd}
\end{equation*}
in which $\cat{(Mon)OTC^{\otcplus}} \simeq \cat{(Mon)OpCatPar}$ is in fact an isomorphism of categories. These restrict to equivalences between the full subcategories of OTCs with direct sums satisfying Axioms~\ref{axiom:positivity} and \ref{axiom:obsdetops}, and effectuses in partial and total form, respectively. 
\end{theorem}

Finally, combining Theorems~\ref{thm:FreeOTCPlus} and~\ref{thm:1catEquivalence}, we have the following.

\begin{corollary} \label{thm:1catOTCOpCatAdjunction} There is an adjunction 
\begin{equation*} 
\begin{tikzcd}[bend angle = 15, row sep = large]
 \OTC \arrow[rr, bend left, "\OpCat(-)"] & { \ \ \ \   \ \bot} &  \cat{(Mon)OpCat}_{\text{strict}} \arrow[ll, bend left, "\OTCFrom{\Partial(-)}"] 
 \end{tikzcd}
\end{equation*}
between $\cat{(Mon)OTC}$ and the category $\cat{(Mon)OpCat}_{\text{strict}}$ of (monoidal) operational categories coming with specified coproducts, and functors $F \colon \catB \to \catB'$ which preserve them strictly. This restricts to an adjunction between the respective full subcategories of OTCs satisfying Axioms~\ref{axiom:positivity} and \ref{axiom:combiningops} and effectuses in total form.  
\end{corollary}

\section{Discussion}

In this work we introduced operational theories with control, as structures describing the experiments one may perform in a given domain of physics, and argued that they are best understood using the notion of an operational category. We saw that, starting from any theory $\Theta$, its category $\OpCat(\Theta)$ forms an operational category, with the coproducts and terminal object providing an elegant description of the flow of classical data during tests in $\Theta$. 

Alternatively, any operational category $\catB$ may be viewed in partial form $\catC = \Partial(\catB)$, and seen as the category of events of an OTC with direct sums $\OTCFrom{\catC}$, and conversely every such theory arises in this way. In particular, studying $\OpCat(\Theta)$ is equivalent to studying the completion $\Theta^{\otcplus}$ of our theory under direct sums, with category of events $\ParOpCat(\Theta) \simeq \Partial(\catB)$. 

As a special case, we saw that effectuses may be identified with OTCs with direct sums satisfying our Axioms~\ref{axiom:positivity} and \ref{axiom:obsdetops}.

\subsubsection*{Comparison with operational-probabilistic theories}

Our operational theories with control (OTCs) are based on the operational-probabilistic theories (OPTs) of~\cite{PhysRevA.84.012311InfoDerivQT}, and as such, the majority of their results and proofs carry over immediately into any OTC, and hence any (separated) operational category. OPTs differ by not assuming causality, but are otherwise less general. Firstly, Chiribella et al.\ only consider \operations{} of the form $\{f_x \colon A \eventto B\}_{x \in X}$, without varying output systems, though note it is natural to consider more general \operations{}, particularly in causal theories \cite{chiribella2010purification}.  More crucially, along with separation (see Section~\ref{sec:AdditionalAssumptions}), OPTs come with several extra assumptions typical to probabilistic approaches:

\begin{itemize}
\item Scalars $p \colon I \eventto I$ correspond to actual probabilities $p \in [0,1]$. Along with separation, this allows events $f \colon A \eventto B$ to be described by positive maps between ordered vector spaces.
\item Each such space of maps $\Events(A,B)$ is taken to be closed under pointwise limits, on operational grounds.  
\item Each space of states $\Events(I, A)$ is taken to be finite-dimensional.
\end{itemize}

It would be desirable to find further categorical properties which one may add to the definition of an operational category which ensure that each of the above hold, hence providing the full reasoning power of OPTs within a purely categorical framework.

\subsubsection*{Comparison with categorical quantum mechanics}

The framework of categorical quantum mechanics (CQM) due to Abramsky and Coecke \cite{abramskycoecke:categoricalsemantics} resembles our approach. Both study physical theories such as quantum theory abstractly, using symmetric monoidal categories in which arrows $f \colon A \to B$ are interpreted as physical processes. However, there are two main differences. 

Firstly, CQM models quantum theory using the category $\cat{FHilb}$ of finite-dimensional Hilbert spaces and linear maps, which lacks discarding due to the `No-deleting' theorem \cite{pati2000impossibility}. In contrast, since we include classical systems in our description, our basic example of a `quantum' operational category $\FDimCStar$ instead takes (finite-dimensional) C*-algebras as its objects, and includes `mixed' states and processes. Secondly, our approach only models the subcategory of (sub)unital maps, while in CQM one considers arbitrary completely positive maps. This is closely related to the fact that our categories only come with a partial addition $\ovee$, rather than a total one. In future work, we plan to describe a construction which allows one to pass from any suitable operational category, coming with such a partial addition $\ovee$, to one of supernormalised maps, in which $\ovee$ then becomes total.

\subsubsection*{Biproducts versus coproducts }

In Section~\ref{subsec:Direct Sums} we used the terminology of `direct sums' to describe the coproducts in the category $\Partial(\catB)$. Direct sums of spaces are more commonly understood categorically using \emph{biproducts}. In fact, our coproducts closely resemble biproducts in several ways: they induce a partial addition $f \ovee g$ on morphisms in $\Partial(\catB)$ in just the same way as biproducts induce a total addition, and their projections $\triangleright_i$ satisfy the same set of equations as those of a biproduct (Lemma~\ref{OTCPlus_Coproduct_iff}). Our completion of an OTC $\Theta$ to one with direct sums $\Theta^{\otcplus}$ resembles the completion of any semiadditive category to one coming with biproducts. Biproduct structures were used by Abramsky and Coecke in their original paper \cite{abramskycoecke:categoricalsemantics} to model classical data, similarly to our use of coproducts. However, since these biproducts take place in the category $\cat{FHilb}$, the addition on morphisms $f + g$ they induce models quantum \emph{superpositions}, rather than simply coarse-graining.

\subsubsection*{A categorical semantics for operational theories}

Effectus theory has been developed explicitly as a categorical logic for use in the modelling of quantum computation. Our results make this precise, by demonstrating that effectuses, or more generally operational categories, have as their `internal logic' the language of operational theories with control. 

This is akin to the well-known correspondences between intuitionistic logic and topoi, or between models of the simply typed $\lambda$-calculus and Cartesian closed categories. The latter, `Curry-Howard-Lambek', correspondence provides a foundation for functional programming languages such as Haskell \cite{abramsky2011introduction}, and similarly, one may hope that operational categories can be used as the foundation for a programming language suitable for \emph{general probabilistic computation}. Indeed, effectus theory is already being explored as the basis for a quantum programming language by Adams \cite{adams2014qpel}. 

Practically, this correspondence allows one to prove results about operational categories by reasoning using operational theories, and vice versa.  In one direction, the perspective of effectus theory, based on categorical logic, can help clarify ideas in operational physics. For example, in effectus theory, effects $p \colon A \to 1 + 1$ on a physical system are typically referred to as \emph{predicates}. The analogy is with classical logic, described by the effectus $\Set$, in which predicates correspond precisely to characteristic functions $p \colon A \to 1 + 1 = \{0,1\}$, and hence subsets $P \subseteq A$. In this way, one can identify in what ways `operational logic' is really like classical logic, and in what ways it differs.

Conversely, ideas from operational physics, such as those in the quantum reconstruction theorem of \cite{PhysRevA.84.012311InfoDerivQT}, may now be applied to effectus theory. One may hope for a rich interplay between operational ideas and the universal properties studied in categorical logic, just as topos theory has benefited from its connections with logic and geometry. For example, in \cite{QC2015ChoJaWW}, intriguing first steps are taken towards understanding the process of measurement through chains of adjunctions.


\subsection*{Acknowledgements}
Many thanks to Chris Heunen for interesting discussions and feedback on this work, and to Bart Jacobs, who suggested the adjunction of Corollary~\ref{thm:1catOTCOpCatAdjunction}, and allowed my visit to the Institute of Computing at Radboud University Nijmegen during September 2015, where some of these ideas were developed. This work benefited further from discussions with Aleks Kissinger, Kenta Cho, and Bas and Bram Westerbaan, and was supported by EPSRC Studentship OUCL/2014/SET.

\bibliographystyle{eptcs_dphiltweak}
\bibliography{main_v18}

\appendix

\section{Basic Notions from Category Theory} \label{appendix:cattheory}

Here we provide a quick introduction to some of the basic categorical notions used above. The classic text on category theory is \cite{mac1978categories}, while \cite{coecke2011categories} provides an introduction aimed at physicists, including an introduction to symmetric monoidal categories. 

We say that an arrow $f \colon A \to B$ in a category is \emph{monic} when $f \circ g = f \circ h \implies g = h$, for any pair of arrows $g, h \colon C \to A$.  More strongly, a morphism $f \colon A \to B$ is an \emph{isomorphism} when there is some (necessarily unique) arrow $f^{-1} \colon B \to A$ satisfying $f^{-1} \circ f = \id{A}$ and $f \circ f^{-1} = \id{B}$. 
An object $1$ is \emph{terminal} when  every object $A$ has a unique arrow $\unique \colon A \to 1$. Similarly, an object $0$ is \emph{initial} when there is always a unique arrow $\unique \colon 0 \to A$. 

In any category, a \emph{coproduct} of a collection of objects $\{A_x\}_{x \in X}$ is given by an object $\coprod_{x \in X} A_x$ and morphisms $\coproj_y \colon A_y \to \coprod_{x \in X} A_x$, called \emph{coprojections}, such that, for each collection of morphisms $\{f_x \colon A_x \to B \mid x \in X \}$, there is a unique arrow $f \colon \coprod_{x \in X} A_x \to B $ satisfying $f \circ \coproj_x = f_x$ for all $x \in X$. A category has coproducts for all such finite collections $X$ as long as it has an initial object $0$ and coproducts of pairs of objects, which are denoted by $A + B$. In this case, one has $\coprod^n_{i = 1} A_i \simeq A_1 + \ldots + A_n$. For any pair of arrows $f \colon A \to C$, $g \colon B \to C$, we write $[f,g] \colon A + B \to C$ for the unique arrow satisfying $[f,g] \circ \coproj_1 = f$ and $[f,g] \circ \coproj_2 = g$. Given arrows $f \colon A \to B$ and $g \colon C \to D$, we denote by $f + g \colon A + C \to B + D$ the unique arrow with $(f + g) \circ \coproj_1 = \coproj_1 \circ f$ and $(f + g) \circ \coproj_2 = \coproj_2 \circ g$. We define the morphism $f_1 + \ldots + f_n  \colon A_1 + \ldots + A_n \to B_1 + \ldots + B_n$ similarly, given $f_i \colon A_i \to B_i$ for $i=1, \ldots, n$.
For any $n$ and object $A$, the coproduct $\coprod^n_{i = 1} A = \underbrace{A + \ldots + A}_n$ is called the $n$-th \emph{copower} of $A$, denoted $n \cdot A$. It comes with a \emph{codiagonal} morphism $\triangledown \colon n \cdot A \to A$, defined by $\triangledown \circ \coproj_i = \id{}$, for all $i$.

Finally, a \emph{pullback} of a pair of arrows $f \colon A \to C$, $g \colon B \to C$ is given by an object $P$ and arrows $p \colon P \to A$, $q \colon P \to B$ with the following property: for every pair of arrows $h \colon Q \to A$, $k \colon Q \to B$ satisfying $g \circ k = f \circ h$, there is a unique arrow $x \colon Q \to P$ with $k = q \circ x$ and $h = p \circ x$.
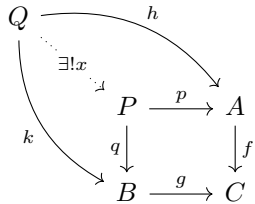
\[
\begin{tikzcd}
Q
\arrow[drr, bend left, "h"]
\arrow[ddr, bend right, swap, "k"]
\arrow[dr, dotted, "\exists \unique {x}" description] & & \\
& P \arrow[r, "p"] \arrow[d, "q", swap]
& A \arrow[d, "f"] \\
& B \arrow[r, "g"]
& C
\end{tikzcd}
\]

\section{Properties of Operational Categories} \label{appendix:propertiesOfOpCats}

We now prove some of the properties of operational categories mentioned in the main text. \\

\noindent \emph{Proof of Lemma~\ref{opcatlemma}.} The coprojection ${\coproj}^{A,B}_1:A \to A + B$ satisfies $(\id{A} + \unique) \circ {\coproj_1}^{A,B} = \coproj_1 \colon A \to A + 1$, which is monic due to the pullback in the definition of an operational category. Hence so is ${\coproj}^{A,B}_1$. 

For the left-hand pullback, consider morphisms $f: A \to B$, $g: D \to B$, $h: D \to A + C$ for which $(f + \id{}) \circ h = \coproj_1 \circ g: D \to B + C$. Then letting $k = (\id{} + \unique) \circ h: D \to A + 1$, it's easy to see that $(\unique + \unique) \circ k = \coproj_1 \circ \unique: D \to 1 + 1$, and so by the pullback in the definition of an operational category,  there is a unique $r: D \to A$ such that $k = \coproj_1 \circ r$. Working in the category $\Partial(\catB)$ we have $\triangleright_1 \partialcomp h = \triangleright_1 \partialcomp k = \triangleright_1 \partialcomp \coproj_1 \partialcomp r = r$ and $\triangleright_2 \partialcomp h = \triangleright_2 \partialcomp (f + \id{}) \partialcomp h = \triangleright_2 \partialcomp \coproj_1 \partialcomp g = 0 = \triangleright_2 \partialcomp \coproj_1 \partialcomp r$.
By joint monicity of the $\triangleright_i$ (Lemma~\ref{lemma:opcatpartialjointmon}), we conclude that $h = \coproj_1 \circ r$ in $\catB$. Next note that, in $\catB$,
$\coproj_1 \circ g = (f + \id{}) \circ \coproj_1 \circ r = \coproj_1 \circ f \circ r$. Since $\coproj_1$ is monic, $g = f \circ r$, and this $r$ is unique, as required. 

The right-hand pullback is in fact a special case of the left-hand one:
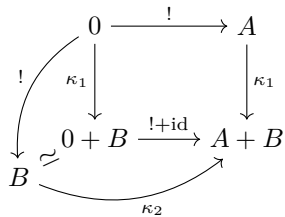
\[
\begin{tikzpicture}[commutative diagrams/every diagram]
\node (tl) at (0,0) {$0$};
\node (bl) at (0,-1.5) {$0 + B$};
\node[rotate=25] (bbl) at (-0.6, -1.75) {$\simeq$};
\node (bbl) at (-1,-2){$B$};
\node (tr) at (2,0) {$A$}; 
\node (br) at (2,-1.5) {$A + B$};
\path[commutative diagrams/.cd, every arrow, every label]
(tl) edge node{$\unique$} (tr)
(tl) edge[bend right] node[swap]{$\unique$} (bbl)
(tl) edge node[swap]{$\coproj_1$} (bl)
(bl) edge node{$\unique + \id{}$} (br)
(tr) edge node{$\coproj_1$} (br)
(bbl) edge[bend right, swap] node{$\coproj_2$} (br);
\end{tikzpicture}
\]
\hfill\(\qed\)

Next we turn to the positivity Axiom from Section~\ref{sec:AdditionalAssumptions}.

\begin{lemma}
 If $(\catC, I, \discard{})$ and $(\catC, I, \discard{}')$ are positive operational categories in partial form, then $\discard{A} = \discard{A}'$ for all objects $A$.
 \end{lemma}

\begin{proof}
Note that coarse-graining and compatibility of events are identical in the OTCs defined by $(\catC, I, \discard{})$ and $(\catC, I, \discard{}')$.  Let $a, b \colon A \to I$ be such that $\discard{A}' \ovee a = \discard{A}$, and $\discard{A} \ovee b = \discard{A}'$. Then $\{\discard{A}, a , b \}$ are all compatible, and so $a \ovee b = 0$. By positivity $a = b = 0$ and so $\discard{A} = \discard{A}'$. 
\end{proof}

\begin{lemma} \label{lemma:PosOpCat} Let $\catB$ be a positive operational category. Then: 
\begin{enumerate}[i)] \itemsep0pt
\item \label{anyisototal} Any isomorphism in $\Partial(\catB)$ is total. 
\item \label{initobstrict} The initial object $0$ is strict in $\catB$. That is, any morphism $f \colon A \to 0$ is an isomorphism.
\item \label{posopcatpull} Diagrams of the form \eqref{eq:generalPosPullbacks} are pullbacks in $\catB$.
\end{enumerate}
\end{lemma}
\begin{proof}
For the first two parts, we reason in the theory $\OTCFrom{\Partial(\catB)}$. \begin{enumerate}[i)]
\item
Let the event $f \colon A \eventto B$ be an isomorphism, with \operations{} $\{f, e\}$ and $\{f^{-1}, e' \}$ for (unique) effects $e$, $e'$. Then by control $\{f^{-1} \eventcomp f, e' \eventcomp f, e\}$ is \anoperation{}, and since $f^{-1} \eventcomp f = \id{A}$ is deterministic, by positivity $e = 0$. Hence $f$ is deterministic also, i.e.\ total.

\item
 If $f \colon A \eventto 0$ is a deterministic event, then $\discard{A} \eventcomp \id{A} = \discard{A} = \discard{0} \eventcomp f = \discard{0} \eventcomp f = 0$, and so $\id{A} = 0_{A,A}$ by positivity. It follows that $A \simeq 0$ and,since both objects are initial, that $f$ is an isomorphism.
\item 
Both the right-hand and outer rectangles in the diagram below are pullbacks.
\[
\begin{tikzcd}
A \rar{f} \dar[swap]{\coproj_1} & B \dar{\coproj_1} \rar{\unique} & 1 \dar{\coproj_1} \\
A + C \rar[swap]{f + g} & B + D \rar[swap]{\unique + \unique} & 1 + 1
\end{tikzcd}
\]
By a well-known result, the `Pullback Lemma', this means that the left-hand square is also.
\end{enumerate} \end{proof}

\section{Proofs of Functorial Results} \label{appendix:omittedproofs}

Let us now prove the results of Section~\ref{sec:CategoricalResults}, establishing the functorial correspondences between operational theories and categories. 

\subsection{Categorical description of the direct sum completion}

\noindent \emph{Proof of Lemma~\ref{lemma:morphismOTC}.}  Let $F \colon \Theta \to \Theta'$ be a morphism of OTCs. Since $F$ preserves \operations{}, $F(\discard{A})$ is deterministic, and so we have the desired equality  $\discard{F(A)} = \discard{F(I)} \circ F(\discard{A})$. 
For the second part, we will show that $F(0_{A,I}) = 0_{F(A), F(I)}$ for all systems $A$, and then we always have $F(0_{A,B}) = F(0_{I,B}) \circ F(0_{A,I}) = 0_{F(A), F(B)}$, as required.  
As $F$ preserves \operations{}, $\{F(0), F(\discard{A})\}$ forms \anoperation{} in $\Theta'$. Applying $\discard{F(I)}$, along with the first part, gives $\discard{F(I)} \circ F(0_{A,I}) = {\discard{F(A)}}^{\bot} = 0_{F(A),I}$. Composing with $\discard{F(I)}^{-1}$ we then have $F(0_{A,I}) = 0$. Since $F$ preserves \operations{}, coarse-graining, identities and $0$, it preserves direct sums by Lemma~\ref{OTCPlus_Coproduct_iff}. \hfill\(\qed\)

\phantom \\ 

\noindent \emph{Proof of Theorem~\ref{thm:FreeOTCPlus}.}
We consider $\Theta^{\otcplus}$ with the `obvious' choice of direct sums $\bigoplus_{x \in X} \indsys{A_y}_{y \in Y_x} = \indsys{A_y}_{x \in X, y \in Y_x} $. 
For any (monoidal) OTC $\Theta$, let $\eta \colon \Theta \to \Theta^{\otcplus}$  be the embedding morphism of (monoidal) OTCs given by  $\eta(A) = \indsys{A}$ and $\eta(f \colon A \eventto B) = \indsys{ f \colon A \eventto B}$. 
Then for any OTC $\Theta'$ coming with a specified direct sum structure, any morphism of (monoidal) OTCs $F \colon \Theta \to \Theta'$ has a unique extension to an arrow $\bar{F} \colon (\Theta^{\otcplus}, \oplus) \to (\Theta', \oplus)$ in $\OTCPlusStrict$ satisfying $U \bar{F} \circ \eta = F$, defined as follows: we set $\bar{F}(\indsys{A_x}_{x \in X}) = \bigoplus_{x \in X} F(A_x)$ and for $M \colon \indsys{A_x}_{x \in X} \eventto \indsys{B_y}_{y \in Y}$ define $\bar{F}(M)$ to be the unique event $\bigoplus_{x \in X} F(A_x) \eventto \bigoplus_{y \in Y} F(B_y)$  satisfying $\triangleright_y \eventcomp \bar{F}(M) \eventcomp \coproj_x = M(x,y)$ for all $x \in X$, $y \in Y$. 
\[ 
\begin{tikzcd} 
\Theta \rar{F} \dar[hook, swap]{\eta} & \Theta' \\
U(\Theta^+) \arrow[ur, swap, dotted, "U(\bar{F})"] & 
\end{tikzcd}
\]
By standard categorical results, this ensures that $(-)^{\otcplus}$ extends to a left adjoint to $U$.
\hfill\(\qed\) \\

It would be more natural to describe the $\Theta^{\otcplus}$ construction in terms of $\cat{(Mon)OTC}^{\otcplus}$, without requiring strict preservation of direct sums. To do so, one needs to view $\cat{(Mon)OTC}$ and $\cat{(Mon)OTC}^{\otcplus}$ as \emph{2-categories}, with `arrows between arrows', called 2-cells. Each of our categories $\cat{(Mon)OTC}$,  $\cat{(Mon)OTC}^{\otcplus}$, $\cat{(Mon)OpCatPar}$ and $\cat{(Mon)OpCat}$ form strict 2-categories with 2-cells $\alpha \colon F \Rightarrow G$ given by (monoidal) natural transformations. The direct sum completion $(-)^{\otcplus}$ then in fact forms a left \emph{bi-adjoint} to the strict 2-functor $U \colon \cat{(Mon)OTC}^{\otcplus} \to \cat{(Mon)OTC}$.

\subsection{Categorical equivalence of operational theories and categories}

We now wish to prove Theorem~\ref{thm:1catEquivalence}, which describes the equivalence of categories $\cat{(Mon)OTC^{\otcplus}} \simeq \cat{(Mon)OpCatPar} \simeq \cat{(Mon)OpCat}$. In fact, a stronger statement holds: the equivalence extends to 2-cells, making it an equivalence of strict 2-categories. First, we consider the isomorphism $\cat{(Mon)OTC^{\otcplus}} \simeq \cat{(Mon)OpCatPar}$.

\begin{lemma} The assignments $\Theta \mapsto \Events_{\Theta}$ and $\catC \mapsto \OTCFrom{\catC}$ define an isomorphism of (strict 2-)categories $\cat{(Mon)OTC^{\otcplus}} \simeq \cat{(Mon)OpCatPar}$. 
\end{lemma}
\begin{proof}
The above assignment is a bijection on objects by Theorems \ref{thm:opcattopartialopcat} and \ref{thm:paropcatToOpCat}. Let $F: \Events_{\Theta} \to \Events_{\Theta'}$ be a morphism of OTCs with direct sums. 
By Lemma~\ref{defmorphismOTC}, $F$ preserves discarding and direct sums, and hence by Lemma~\ref{OTCPlus_Coproduct_iff} also preserves coproducts.
Conversely, suppose $F$ is a morphism in $\cat{(Mon)OpCatPar}$. Then whenever $f \colon A \to B$ is total, so is $F(f)$, since:
\[
\discard{F(B)} \circ F(f) = \discard{F(I)} \circ F(\discard{B}) \circ F(f)  = \discard{F(I)} \circ F(\discard{A}) = \discard{F(A)}
\]
as $F$ preserves discarding. By Corollary~\ref{thm:OTCDSumstoopcat}, \operations{} and coarse-graining are defined in $\Theta$ and $\Theta'$ using totality and the coproducts, which are then preserved by $F$. Hence $\cat{(Mon)OTC^{\otcplus}}$ and $\cat{(Mon)OpCatPar}$ have the same morphisms. Clearly they also have the same 2-cells.
\end{proof}

\subsection{Categorical equivalence of operational categories in total and partial form}

Finally, we wish to show that $\cat{(Mon)OpCatPar}$ and $\cat{(Mon)OpCat}$ are equivalent as categories (and in fact also as strict 2-categories). Most of the work has already been done by Cho, who in \cite{Partial2015Cho}, establishes a 2-categorical equivalence between effectuses in total and partial form. In fact, the same proof establishes a more general result, applicable in particular to (monoidal) operational categories.

We define (strict 2-)categories $\cat{(Mon)Total}$, $\cat{(Mon)Partial}$ as follows.
An object of $\cat{Total}$ is a category $\catB$ coming with coproducts and a terminal object $1$. An object of $\cat{MonTotal}$ is in addition symmetric monoidal, with $1$ as its tensor unit, and the tensor distributing over coproducts. Similarly, an object $(\catC, \discard{})$ of $\cat{(Mon)Partial}$ is defined just like a (monoidal) operational category in partial form (Definition~\ref{def:opcatpartial}), but without conditions~\ref{paropcatcatJM}) or \ref{paropcatuniquetotal}). Morphisms and 2-cells in each case are the same as in $\cat{(Mon)OpCat}$ and $\cat{(Mon)OpCatPar}$, respectively.
\begin{theorem}\label{thm:ChoCorrespondence} The assignments $\catB \to \Partial(\catB)$ and $\catC \to \catC_{\total}$ extend to (strict 2-)functors $\Partial(-) \colon \cat{(Mon)Total} \to \cat{(Mon)Partial}$ and $(-)_{\total}\colon \cat{(Mon)Partial} \to \cat{(Mon)Total}$. Further, there are (strict 2-)natural transformations:
\begin{itemize}
\item $\Phi_{\catB} \colon \catB \to \Partial(\catB)_{\total}$ given by $\Phi_{\catB}(A) = A$ and $\Phi_{\catB}(f) = \totaspar{f}$

\item $\Psi_{\catC} \colon \Partial(\catC_{\total}) \to \catC$ given by $\Psi_{\catC}(A) = A$ and $\Psi_{\catC}(f) = \triangleright_1 \circ f$
\end{itemize}
forming the unit and counit of a (strict 2-)adjunction $\Partial(-) \dashv (-)_{\total}$.
\end{theorem}
\begin{proof}
The details are worked out in \cite{Partial2015Cho} with an emphasis on effectuses, but do not rely on all of the effectus axioms. Here we just sketch the main ideas from Cho's proof.

We have seen that every category $\catB$ in $\cat{(Mon)Total}$ defines a `partial' one $\Partial(\catB)$ in $\cat{(Mon)Partial}$, with discarding maps $\discard{A} = \unique \circ \coproj_1 \colon A \to 1 + 1$. Any functor $F \colon \catB \to \catB'$ in $\cat{Total}$ may be lifted to a a functor $\lift{F} \colon \Partial(\catB) \to \Partial(\catB')$ defined by 
\[
\lift{F}(
\begin{tikzcd}[column sep = small]
A \rar{f} & {B+1}
\end{tikzcd}
)
=
\big( 
\begin{tikzcd}
F(A) \rar{F(f)} & F(B + 1) \rar{\simeq} & F(B) + 1
\end{tikzcd}
\big)
\]
This lifting on 1-cells is compatible with the functor $\totaspar{-}$ in that 
\[
\begin{tikzcd}
\catB \rar{F} \dar[swap]{\totaspar{-}} & \catB' \dar{\totaspar{-}} \\
\Partial(\catB) \rar[swap]{\lift{F}} & \Partial(\catB')
\end{tikzcd}
\]
commutes. Further, any natural transformation $\alpha \colon F \Rightarrow F'$ between such functors lifts to one $\lift{\alpha} \colon \lift{F} \Rightarrow \lift{F'}$ given by $\lift{\alpha}_A = \coproj_1 \circ \alpha_A \colon F(A) \to G(A) + 1$.
When $\catB$ is symmetric monoidal with $\otimes$ distributing over $+$, $\lift{F}$ is a symmetric strong monoidal functor, and $\lift{\alpha}$ is a monoidal natural transformation.  Further, $\lift{F}$ then preserves discarding and so is indeed a morphism in $\cat{(Mon)Partial}$. In this way, one obtains a strict 2-functor $\Partial(-) \colon \cat{(Mon)Total} \to \cat{(Mon)Partial}$. 

Conversely, each object $\catC$ of $\cat{(Mon)Partial}$ defines a new `total' one $\catC_{\total}$ in $\cat{(Mon)Total}$.
The 2-functor $(-)_{\total} \colon \cat{(Mon)OpCatPar} \to \cat{OpCat}$ sends each morphism $G \colon \catC \to \catC'$ to its restriction $G|_{\total} \colon \catC_{\total} \to \catC'_{\total}$. Since any such $G$ preserves discarding, it preserves totality of arrows, and so $G_{\total}$ is well-defined and preserves the terminal object. The coproducts in $\catC$ restrict to $\catC_{\total}$ and hence are also preserved by $G|_{\total}$.  In the monoidal case, we claim that the coherence morphisms $u \colon I \to F(I)$ and $\phi_{A,B} \colon G(A) \otimes G(B) \to G(A \otimes B)$ for $G$ are always total. We have $u = \discard{}^{-1}$, and so $u$ is indeed total. Then
\begin{align*}
\discard{G(A \otimes B)} \circ \phi_{A,B} & = \discard{G(I)} \circ G(\lambda_I) \circ G(\discard{A} \otimes \discard{B}) \circ \phi_{A,B} \\
& = u^{-1} \circ G(\lambda_I) \circ \phi_{A,B} \circ (G(\discard{A}) \otimes G(\discard{B})) \\
& = \lambda_I \circ (u^{-1} \otimes u^{-1}) \circ (G(\lambda) \otimes G(\lambda)) \\
& = \lambda \circ (\discard{G(A)} \otimes \discard{G(B)}) \\
& = \discard{G(A) \otimes G(B)}
\end{align*}
where in the third step we used that any symmetric monoidal functor $(G,u,\lambda)$ satisfies $u \circ \lambda = G(\lambda) \circ \phi \circ (u \otimes u)$. Since the tensor $\otimes$ restricts to $\catC_{\total}$, $G|_{\total}$ is again a symmetric monoidal functor. 

Next, one may check that for any 2-cell $\alpha \colon F \Rightarrow G $ in $\cat{Partial}$ each $\alpha_A \colon F(A) \to G(A)$ is total, and so $\alpha$ restricts to a 2-cell $F|_{\total} \Rightarrow G|_{\total}$  in $\cat{(Mon)Total}$. In this way we obtain a strict 2-functor $(-)_{\total} \colon \cat{(Mon)Partial} \to \cat{(Mon)Total}$. 

It's straightforward to verify that each $\Phi_{\catB} \colon \catB \to \Partial(\catB)_{\total}$ defined as above is indeed a morphism in $\cat{(Mon)Total}$, and each $\Psi_{\catC} \colon \catC \to \Partial(\catC_{\total})$ is a morphism in $\cat{(Mon)Partial}$, and that both assignments are indeed strictly 2-natural. Finally, we check that $\Psi$ and $\Phi$ satisfy the triangle identities. For each $\catC$ in $\cat{(Mon)Partial}$ and total $f \colon A \to B$ in $\catC$, we have: 
\begin{align*}
\Psi_{\catC_{\total}} \circ \Phi_{\catC_{\total}}(f) = \Psi_{\catC_{\total}}(\totaspar{f}) 
= \triangleright_1 \partialcomp \totaspar{f} = \triangleright_1 \circ \coproj_1 \circ f = f
\end{align*}
and so $\Psi_{\catC_{\total}} \circ \Phi_{\catC_{\total}} = \id{\catC_{\total}}$. 
Similarly, for each $\catB$ in $\cat{(Mon)Total}$ and $g \colon A \partialto B$ in $\Partial(\catB)$, we get that:
\begin{align*}
\Psi_{\Partial(\catB)} \circ \Partial(\eta_{\catB})(g) = \triangleright_1 \partialcomp \coproj_1 \partialcomp g = g
\end{align*}
giving $\Psi_{\Partial(\catB)} \circ \Partial(\Phi_{\catB}) = \id{\Partial(\catB)}$. 
\end{proof}

\begin{corollary} \label{cor:2CatEquivOpCat} There is a (strict 2-)equivalence of (strict 2-)categories $\Partial(-) \colon \cat{(Mon)OpCat} \to \cat{(Mon)OpCatPar}$, $(-)_{\total} \colon \cat{(Mon)OpCatPar} \to \cat{(Mon)OpCat}$.
\end{corollary}
\begin{proof}
By Theorems~\ref{thm:opcattopartialopcat} and~\ref{thm:paropcatToOpCat}, the strict 2-functors and natural transformations of Theorem~\ref{thm:ChoCorrespondence} restrict to $\cat{(Mon)OpCat} \leftrightarrow \cat{(Mon)OpCatPar}$. By Theorem~\ref{thm:opparopcatEquivalence}, each of the components $\Phi_{\catB} \colon \catB \to \Partial(\catB)_{\total}$ and $\Psi_{\catC} \colon \Partial(\catC)_{\total} \to \catC$ are then isomorphisms.
\end{proof}

\end{document}